\documentclass[twocolumn,showpacs,preprintnumbers,amsmath,amssymb]{revtex4-2}

\usepackage{dcolumn}
\usepackage{bm}
\usepackage{graphicx}
\usepackage{amsmath}
\usepackage{amsfonts}
\usepackage{amssymb}
\usepackage{amsthm}
\usepackage{amstext}
\usepackage{amsbsy}
\usepackage{amsopn}
\usepackage{amscd}
\usepackage{amsxtra}
\usepackage{hyperref}
\usepackage{empheq}
\usepackage{color}
\usepackage{soul}

\usepackage{enumitem} 
\usepackage{comment}

\newtheorem{theorem}{Theorem}
\newtheorem{lemma}{Lemma}

\def\qed{$\Box$}
\def\phi{\varphi}
\def\epsilon{\varepsilon}

\newcommand{\be}{\begin{equation}}  
\newcommand{\ee}{\end{equation}}
\newcommand{\ba}{\begin{array}}
\newcommand{\ea}{\end{array}}
\newcommand{\bea}{\begin{eqnarray}}
\newcommand{\eea}{\end{eqnarray}}
\newcommand{\bra}{\langle}
\newcommand{\ket}{\rangle}

\newcommand{\nn}{\nonumber}

\everymath{\displaystyle}

\begin{document}
\title{Entropy production versus memory effects in two-level open quantum systems} 
\author{G. Th\'eret, D. Sugny and C. L. Latune}
\email{Camille.Lombard-Latune@u-bourgogne.fr}
\affiliation{Université Bourgogne Europe, CNRS, Laboratoire Interdisciplinaire Carnot de Bourgogne ICB UMR 6303, 21000 Dijon, France}

\begin{abstract}
We compare several definitions of entropy production rate introduced in the literature from a large variety of situations and motivations, and then analyze their relations with memory effects. Considering a relevant experimental example of a qubit interacting with a single bosonic mode playing the role of a finite bath, we show that all definitions of entropy production coincide at weak coupling. In a stronger coupling regime,  significant discrepancies emerge between the different entropy production rates, although some similarities in the overall behaviour remain. However, surprisingly, two of these definitions—one based on local quantities of the system and the other on non-local quantities—coincide exactly, even in the case of strong coupling. Finally, a high degree of correspondence is observed when memory effects characterized by P-divisibility are compared with the sign of all entropy production rates in the case of weak coupling. Such correspondence degrades at stronger coupling, leading us to extend the concept of entropy production to the dynamical map. We show a perfect equivalence between the sign of this enlarged concept of entropy production and P-divisibility, both numerically and analytically, in the case of phase-covariant master equations.
\end{abstract}

\maketitle

\section{Introduction}
The concept of entropy production is central to both classical and quantum thermodynamics~\cite{Landi2021}. It is closely related to the irreversibility of the local system dynamics and to the information lost in the bath~\cite{Esposito2010,Landi2021,Santos2018}, and is also associated with the arrow of time~\cite{Parrondo2009,Deffner2011,Santos2019,Latune2020}.
From a practical viewpoint, the entropy production also determines and often limits our ability to perform operations and tasks such as cooling a system~\cite{Taranto2023}, extracting work~\cite{Uzdin2015}, erasing information~\cite{Vanvu2022, Rolandi2023}, measuring time~\cite{Milburn2020, Pearson2021}, and controlling fluctuations~\cite{Hasegawa2020}, to mention a few. 

More recently, several studies have suggested a close relation between entropy production and memory effects~\cite{Matola2012,Bhattacharya2017,Marcantoni2017,Popovic2018,Santos2018, Strasberg2019, Rivas2020Apr}, often refereed to as non-Markovian effects, in reference to classical dynamics. In Ref.~\cite{Strasberg2019}, the authors establish the conditions in classical systems for which a negative entropy production rate implies non-Markovianity, and raises the question of its quantum version. Indeed, one of the difficulties in the quantum case is the non-unique characterization of non-Markovianity, which stems from the inability to extend the classical notion straightforwardly to quantum systems~\cite{Rivas2014}.  
For open quantum dynamics, the intuitive relation between entropy production and memory effects can be traced back to a measure of memory effects such as the BLP criterion~\cite{Breuer2009}, which relies on the contractivity of the dynamical map. More precisely, when the map is contractive, i.e. when the trace distance between two arbitrary initial states decreases monotonically in time, there is a monotonic loss of information from the system to the bath, which is interpreted as the absence of memory effect. Otherwise, information flows back from the bath to the system. In Ref.~\cite{Matola2012}, the authors provide an upper bound on contractivity in terms of system-bath correlations.

On the other hand, many definitions of entropy production are based on the contractivity of the pseudo-distance (the relative entropy) between an arbitrary initial state and a specific state, such as the steady state (for weak coupling time-independent dissipator)~\cite{Spohn1978, Alicki1979, Breuer2002}, or the instantaneous fixed point (for strong coupling and time-dependent dissipator)~\cite{Colla2022}, or even the joined uncorrelated state of the system and the initial state of the bath~\cite{Esposito2010}.  
This contractivity is analogous to the BLP criterion, and therefore suggests, to explore the precise link between memory effects and entropy production, as studied in~\cite{Bhattacharya2017,Marcantoni2017,Popovic2018,Santos2018}. It can be shown that negative entropy production rates lead to memory effects, and more precisely to Non-Markovianity (defined as non CP-divisibility, see below for a description and~\cite{Breuer2016}). The reverse implication is not true, and potential alternative witnesses of memory effects are introduced such as the purity~\cite{Bhattacharya2017}, and the influence of the form of the instantaneous fixed point of the map~\cite{Marcantoni2017}.  
These previous works also raised the key issue of which definition of entropy production should be used. This problem was studied in~\cite{Colla2021Nov}, where two definitions of entropy production, the traditional one~\cite{Spohn1978, Alicki1979, Breuer2002} and the Esposito definition~\cite{Esposito2010}--tailored for strong coupling between system and bath-- are compared. Interestingly, the two quantities exhibit a very different behaviour, except in the case of very weak coupling and high temperature.

In this paper, we propose to pursue efforts in this direction using a simple but relevant model consisting of a qubit interacting with a single-bosonic-mode bath. Such a model combines strong coupling and finite bath, two characteristics that are currently receiving a lot of attention due to their complexity and close proximity with experimental situations. In the first part of the paper, we compare several definitions of entropy production, ranging from traditional to more recent ones that are specifically designed for strong coupling and finite baths. In a second part, we investigate the relation to memory effects.
The main result is that, at vanishing or small coupling, all definitions of entropy production coincide. However, in a stronger coupling regime, important discrepancies emerge, even though the global behaviour remains similar.  Furthermore, we observe that two definitions of entropy production exactly coincide. Finally, we show numerically the equivalence between P-divisibility and the sign of the entropy production rate of the dynamical map, which is an extension at the map level of the concept of entropy production. This equivalence is then proven mathematically for a large class of open quantum systems governed by a phase-covariant dynamics~\cite{Filippov,Theret2025}.



These results provide a physical interpretation of some of the suggested entropy production definitions and of the concept of P-divisibility. More generally, our study establishes a direct equivalence between the presence of memory effects and the sign of an entropy-production-related quantity. 

The paper is organized as follows. In Sec.~\ref{sec:entprod}, we provide the different definitions of entropy production used in this work. Our illustrative example is introduced in Sec.~\ref{sec:minimalist}. A systematic numerical comparison of entropy production rates is proposed in Sec.~\ref{sec:EPcomparison}. Section~\ref{sec:entropymemory} focuses on the link between entropy production and memory effects. The analytical proof of the equivalence between P-divisibility and the positivity of the map entropy production is presented in Sec.~\ref{sec:proof}. Conclusion and prospective views are given in Sec.~\ref{sec:conclusion}. Additional results are described in the Appendices.

\section{Entropy production}\label{sec:entprod}
In this section, we introduce the main characteristics of the different definitions of entropy production considered in our study. We start with the traditional definition~\cite{Spohn1978, Alicki1979, Breuer2002}, which was introduced for weak coupling between $A$, the quantum system of interest, and, $B$, the thermal bath at temperature $T_B$ (inverse temperature $\beta_B=1/k_BT_B$). According to such a traditional definition, the entropy production is
\be
\Sigma(t) := \Delta S_A - \beta_B Q(t),
\ee
where $\Delta S_A$ is the variation of von Neumann entropy $S_A(t) =S[\rho_A(t)] := -{\rm Tr}[\rho_A(t) \ln \rho_A(t)]$ of the system, and $Q(t)$ is the heat received by $A$, 
\be\label{eq:heatA}
Q(t) := \int_0 ^t du {\rm Tr}[\dot \rho_A(u) H_A(u)].
\ee
The rate of entropy production is given by
\be\label{eqentropy}
\sigma(t) = \dot S_A(t) - \beta_B \dot Q(t),
\ee
where $\dot S_A(t)$ is the time-derivative of von Neumann entropy of the system and $\dot Q(t)$ is the rate of heat exchanged with the bath.
Interestingly, definition~\eqref{eqentropy} of the entropy production rate $\sigma(t)$, inspired from its classical counterpart, coincides with  the formal definition from open quantum system theory~\cite{Spohn1978,Breuer2002},
\be\label{eq:traddef}
\sigma(t) = -\frac{d}{dt} D[\rho_A(t)|\rho_A^\text{ss}],
\ee
where $\rho_A^\text{ss}$ is the steady state of the open dynamics and $D[\sigma|\rho]:= {\rm Tr}[\sigma( \ln \sigma -\ln \rho)]$ is the relative entropy. To be more precise, the coincidence occurs when $\rho_A^\text{ss}$ is equal to the thermal state at the inverse bath temperature $\beta_B$, which is usually the case in the weak coupling regime. Additionally, we can mention that $\sigma(t)$ can also be expressed as the amount of information lost in the bath (still in the weak coupling regime),
\be
\sigma(t) = \dot I_{A:B}(t) + \frac{d}{dt}D[\rho_B(t)|\rho_B(0)],
\ee
where $I_{A:B}(t):= S_A(t) + S_B(t) - S_{AB}(t)$ is the mutual information between $A$ and $B$, so that $\dot I_{A:B}$ corresponds to the rate of creation of correlations between $A$ and $B$.

However, in the strong coupling regime, the definition~\eqref{eq:traddef} is expected to become invalid, partly because the instantaneous thermal state does not always coincide with the instantaneous fixed point (when it exists). More precisely, denoting the dissipator describing the strong coupling dynamics by ${\cal L}_t$, an instantaneous fixed point is defined as a state $\rho^\text{ss}(t)$ which cancels the dissipator at time $t$, meaning ${\cal L}_t \rho^\text{ss}(t) = 0$. By contrast, the instantaneous thermal state is $w[\tilde H_A(t),\beta_B] := e^{-\beta_B \tilde H_A(t)}/{\rm Tr}\left[e^{-\beta_B \tilde H_A(t)}\right]$, with $\beta_B$ the bath temperature and $\tilde H_A(t)$ the Hamiltonian of $A$ appearing in the unitary part of the master equation, see for instance Eq. \eqref{eq:MEingen}. Note that we introduce the following notation, $w[H,\beta]: = e^{-\beta H}/{\rm Tr}[e^{-\beta H}]$, which will be used throughout the paper to design the thermal state at inverse temperature $\beta$ associated with the Hamiltonian $H$. 
To go beyond these limitations,  many definitions of entropy production were suggested in the literature, depending of the context and the chosen viewpoint. We introduce some of them below. A summary is given in Tab.~\ref{tab:prices}.
\begin{itemize}
\item We start with one of the most accepted definition, valid in strong coupling thermal baths, \cite{Esposito2010}: 
\be\label{sigmaEs}
\sigma^{Es}(t) := \dot S_A(t) + \beta_B \dot E_B(t),
\ee
where $\dot E_B(t) = {\rm Tr}[\dot \rho_B(t) H_B]$ is the rate of the energy change of the bath, which is assumed to be initially in a thermal state at inverse temperature $\beta_B$. When the system and the bath are initially uncorrelated, the integral of the rate is guaranteed to be always positive, $\Delta S_A + \beta_B \Delta E_B \geq 0$. However, one limitation is that it only applies when the bath is initially in a thermal state. This implies that $\dot E_B(t)$ is considered to be the heat flow received by the bath. In the strong coupling regime, due to the non-negligible coupling energy, it differs from the heat $Q(t)$ received by $A$ defined in Eq.\eqref{eq:heatA}.

\item In order to extend to more general situations where the bath can be of arbitrary size and start in an arbitrary state, the following definition was introduced~\cite{Elouard2023}, 
\be
\sigma^{El}(t) := \dot S_A(t) + \beta_B^\text{eff}(0) \dot E_B^\text{th}(t),
\ee
where $\dot E_B^\text{th}(t) = {\rm Tr}\left\{H_B\frac{d}{dt}w[H_B,\beta_B^\text{eff}(t)]\right\}$ is the variation rate of the so-called thermal energy of the bath. In the above expression, the effective inverse temperature $\beta_B^\text{eff}(t)$ is defined for all times $t$ as the inverse temperature of the thermal state of the same entropy as $\rho_B(t)$, the state of $B$ at time $t$. In other words, $\beta_B^\text{eff}(t)$ is such that $S\{w[H_B,\beta_B^\text{eff}(t)]\} = S[\rho_B(t)] := -{\rm Tr}[\rho_B(t) \ln \rho_B(t)]$. Note that when $B$ is in a thermal state at temperature $T_B$, we have $\beta_B^\text{eff}(t) = (k_BT_B)^{-1}$.

\item Another framework has been also recently introduced which is valid for arbitrary coupling strength but only for $B$ initially in a thermal state (and uncorrelated from $A$) ~\cite{Colla2022}. However, one technical difficulty is that it is based on an exact master equation describing the local dynamics of $A$. Assume that we do have such a master equation, then it can always be cast into the form,
\be\label{eq:MEingen}
\dot \rho_A(t)  = - \frac{i}{\hbar} [\tilde H_A(t),\rho_A(t)] + {\cal D}_t \rho_A(t),
\ee
where $\tilde H_A(t)$ is in general a time-dependent Hamiltonian different from the free Hamiltonian $H_A$ of $A$, and ${\cal D}_t$ denotes the dissipator of minimal norm (see \cite{Colla2022}) resulting from the interaction with $B$, which is also time-dependent in general. Then, the entropy production is defined as ~\cite{Colla2022}
\bea\label{sigmaco}
&&\sigma^{Co}(t) :=\\
&& -{\rm Tr}\Big\{\dot \rho_A(t) \{\ln\rho_A(t) - \ln w[\tilde H_A(t),\beta_B]\}\Big\},\nn
\eea
where $\beta_B$ is the initial inverse temperature of $B$. Note that in general $w[\tilde H_A(t),\beta_B]$ is not the instantaneous fixed point of the dynamics, meaning it does not satisfy ${\cal D}_t w[\tilde H_A(t),\beta_B] = 0$. As already mentioned, this is one inherent difficulty of strong coupling and finite baths, i.e. the instantaneous thermal state and the instantaneous fixed point no longer coincide. This observation naturally leads to the last definition of entropy production we will consider in our study.

\item We substitute in Eq.~\eqref{sigmaco} the instantaneous thermal state by the instantaneous fixed point of the map:
\be\label{sigmafp}
\sigma^{fp}(t) := -{\rm Tr}\{\dot \rho_A(t) [\ln\rho_A(t) - \ln \rho_A^{fp}(t)]\},
\ee
with $\rho_A^{fp}(t)$ such that ${\cal L}_t \rho_A^{fp}(t) = 0$, where ${\cal L}_t$ is the generator of the open dynamics followed by $A$ (and which will be fully specified in Sec.~\ref{sec:minimalist}). The authors of Ref.~\cite{Colla2022} actually updated their definition \eqref{sigmaco} in a recent proposal \cite{Colla2024} where they choose a time-dependent inverse temperature $\beta_B^\text{fp}(t)$ of $B$ such that $w[H_A(t),\beta_B^\text{fp}(t)]$ is the instantaneous fixed point, recovering the above definition~\eqref{sigmafp} in the case of a quantum harmonic oscillator.

\item Finally, it has also been pointed out that the build up of correlations between system and bath is a central contribution to entropy production~\cite{Esposito2010, Santos2018, Landi2021}, and constitutes a minimal entropy production definition. The correlations between $A$ and $B$ are defined as 
\be
I_{A:B}(t) := S_A(t) + S_B(t) - S_{AB}(t).
\ee
$I_{A:B}(t)$ starts at zero since $A$ and $B$ are assumed to be initially uncorrelated, and then increases with time due to correlations. Since the global evolution of $AB$ is unitary, the rate of correlation build-up is
\be\label{eq:defcorrel}
\dot I_{A:B}(t) = \dot S_A(t) + \dot S_B(t). 
\ee
In the following, we include $\dot I_{A:B}(t)$ in the comparison with the other entropy production definitions.

\end{itemize}

\begin{table}[h]
\centering
\renewcommand{\arraystretch}{1.4}
\begin{tabular}{|c|c|}
\hline
Name & Expression \\
\hline
$\sigma^{Es}$ & $\dot S_A(t) + \beta_B \dot E_B(t) $\\
$\sigma^{El}$ & $\dot S_A(t) + \beta_B^\text{eff}(0) \dot E_B^\text{th}(t)$ \\
$\sigma^{Co}$ &
 $-{\rm Tr}\Big\{\dot \rho_A(t) \{\ln\rho_A(t) - \ln w[\tilde H_A(t),\beta_B]\}\Big\}$\\
$\sigma^{fp}$ & $-{\rm Tr}\{\dot \rho_A(t) [\ln\rho_A(t) - \ln \rho_A^{fp}(t)]\}$ \\
$\dot I_{A:B}$ & $\dot S_A(t) + \dot S_B(t)$\\
\hline
\end{tabular}
\caption{Summary of the different definitions of entropy production considered in the paper as well as the rate of correlations build-up defined in Eq.~\eqref{eq:defcorrel} (see the text for details).}
\label{tab:prices}
\end{table}

\section{Description of the model system}\label{sec:minimalist}
In this section, we present the setting in which we compare the above definitions of entropy production. 
%
%
 We choose a \emph{minimal bath} $B$ composed only of a single bosonic mode and interacting via the Jaynes-Cumming coupling with a two-level system $A$ of free Hamiltonian $H_A=\omega_A \sigma_+\sigma_-$. The total Hamiltonian is therefore
\be
H_{AB} = \omega_A \sigma_+\sigma_- + g (\sigma_+a + \sigma_- a^\dag) + \omega_B a^\dag a,
\ee
with  $\sigma_+ := |1\ket\bra 0|$, $\sigma_-:=|0\ket\bra 1|$, and $|1\ket$, $|0\ket$ denote respectively the excited and ground states of the Pauli matrix $\sigma_z$. We assume that the initial state of $B$ is a thermal state at inverse temperature $\beta_B$, namely $\rho_B(0) =w(H_B,\beta_B)$ and $H_{B} = \omega_B a^\dag a$. 
Note that the Jaynes-Cummings coupling is usually obtained from the Rabi model after the Rotating Wave Approximation (RWA). Therefore, the Jaynes-Cummings coupling is not expected to be valid in the strong coupling regime. However, a different gauge choice could extend its validity to this regime~\cite{Stokes2019}. Despite considering regimes of parameters that present strong coupling characteristics (see e.g. Fig.~\ref{fig:all}), the validity of the RWA is verified by comparing the dynamics under the Rabi model with the Jaynes-Cummings one (see Fig.~\ref{fig:RWA} of Appendix~\ref{app:addplots}). Note that the RWA is only important here for computing the entropy production $\sigma^{Co}$, which requires knowledge of the exact master equation describing the open dynamics. As shown in~\cite{Smirne2010}, this equation can be derived in the case of the Jaynes-Cummings Hamiltonian. 

 Then, under these  conditions, the reduced dynamics of $A$,
\be
\dot \rho_A = {\rm Tr}_B(\dot \rho_{AB}) = -\frac{i}{\hbar} {\rm Tr}_B([H_{AB},\rho_{AB}]),
\ee
can be expressed as an exact master equation~\cite{Smirne2010}:
\bea\label{eq:ME}
\dot \rho_A(t) = {\cal L}_t\rho_A &:=&  -\frac{i}{\hbar} [\Omega_A(t)\sigma_+\sigma_-,\rho_A(t)] \nn\\
&+& \gamma_1(t)\left[\sigma_+\rho_A(t)\sigma_- - \frac{1}{2}\{\sigma_-\sigma_+,\rho_A(t)\}\right]\nn\\
&+& \gamma_2(t)\left[\sigma_-\rho_A(t)\sigma_+ - \frac{1}{2}\{\sigma_+\sigma_-,\rho_A(t)\}\right]\nn\\
&+& \frac{1}{2}\gamma_3(t)\left[\sigma_z\rho_A(t)\sigma_z - \rho_A(t)\right],
\eea
with
\be
\Omega_A(t) = - \Im[\frac{\dot\gamma(t)}{\gamma(t)}],
\ee
\bea
\gamma(t) &=& \sum_{n=0}^\infty \frac{e^{- n \beta_B\omega_B}}{Z}e^{-i\omega_B t}\nn\\
&&\times\left[\cos(\Omega_n t/2) - i \frac{\Delta}{\Omega_n}\sin(\Omega_n t/2)\right]\nn\\
&&\times\left[\cos(\Omega_{n+1} t/2) - i \frac{\Delta}{\Omega_{n+1}}\sin(\Omega_{n+1} t/2)\right],\nn
\eea
\be
\Omega_n = \sqrt{\Delta^2 + 4 g^2n},\nonumber
\ee
\be\label{gamma1}
\gamma_1(t) = \frac{\alpha(t) \dot\beta(t) - \dot\alpha(t)\beta(t) - \dot\beta(t)}{\alpha(t)+\beta(t) -1},
\ee
\be\label{gamma2}
\gamma_2(t) = \frac{\dot\alpha(t)\beta(t) - \alpha(t)\dot\beta(t) - \dot\alpha(t)}{\alpha(t)+\beta(t) -1},
\ee
\be\label{eq:alpha}
\alpha(t) = \sum_{n=0}^\infty \frac{e^{- n \beta_B\omega_B}}{Z}\left[ \cos^2(\Omega_n t/2) + \frac{\Delta^2}{\Omega_n^2}\sin^2(\Omega_n t/2) \right],
\ee 
\bea\label{eq:beta}
\beta(t) &=& \sum_{n=0}^\infty \frac{e^{- n \beta_B\omega_B}}{Z}\nn\\
&&\times\left[ \cos^2(\Omega_{n+1} t/2) + \frac{\Delta^2}{\Omega_{n+1}^2}\sin^2(\Omega_{n+1} t/2) \right],\nn\\
\eea
and
\be
\gamma_3(t) =  -\frac{1}{2}\left[\gamma_1(t)+\gamma_2(t) + 2 \Re[\frac{\dot\gamma(t)}{\gamma(t)}]\right],\nonumber
\ee
with $\Delta=\omega_A-\omega_B$. Note that the quantity $\beta(t)$ defined above and the inverse temperature $\beta_B^\text{eff}(t)$ of $B$ are different quantities. 
 
%
%
%

\section{Entropy production comparison}\label{sec:EPcomparison}
We apply the above definitions of $\sigma^{Es}$, $\sigma^{El}$, $\sigma^{Co}$, $\sigma^{fp}$, and $\dot I_{A:B}$ to the open quantum system introduced in Sec.~\ref{sec:minimalist}. The initial states of $A$ and $B$ are thermal states at inverse temperatures $\beta_A$ and $\beta_B$, respectively.

\begin{figure}
\begin{center}
    \includegraphics[width=0.45\textwidth]{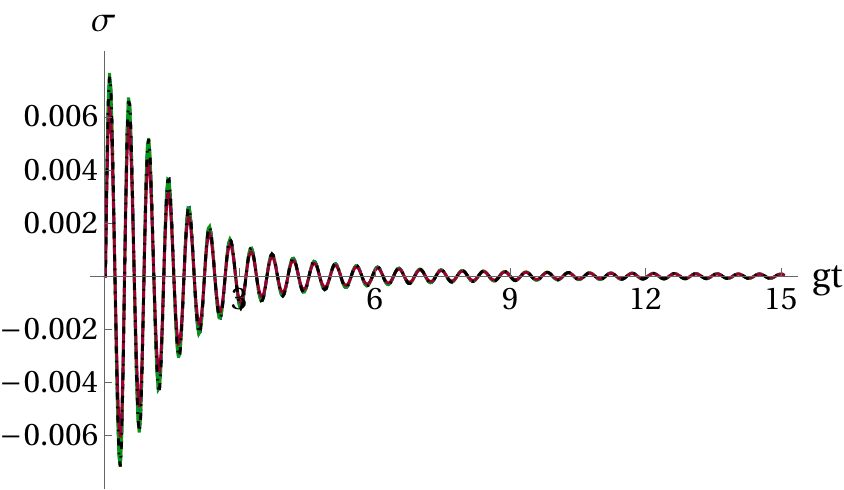}
\end{center}
\caption{Plots of the different definitions of entropy production $\sigma^{Es}$ (orange, thick), $\sigma^{El}$ (green), $\sigma^{Co}$ (purple),  $\sigma^{fp}$ (dashed blue), and $\dot I_{A:B}$ (dot-dashed black) in a weak coupling regime, with $\omega_B/\omega_A = 0.6$, $\Delta/\omega_A = 0.4$, $\omega_A\beta_A = 1.1$, $\omega_A\beta_B = 0.3$, and $g/\omega_A =0.03 $.}\label{fig:all_weak}
\end{figure}

In Fig.~\ref{fig:all_weak}, we present the time evolution of the different definitions of entropy production in a weak coupling situation. 
 Many important observations can be made. Firstly, all entropy production definitions coincide almost exactly (although some small discrepancies start to appear for $gt\gg 1$, not shown here). 
This is an interesting and rather unexpected result because of the large diversity of such definitions. This shows that they are essentially equivalent in the relatively weak coupling regime where $g/\omega_A =0.03 $, even in a finite bath configuration, which emphasizes their relevance.  
Note that some studies~\cite{Breuer1999, Breuer2012} consider that the strong coupling regime starts when the coupling strength is larger than the bath spectral width. In our present model, the bath spectral width is zero, implying that no matter how small is the coupling strength compared to $\omega_A$, it always corresponds to strong coupling according to this criterion.

A second important observation corresponds to the large initial peak in entropy production due to the initial out-of-equilibrium situation ($A$ is initially far from being in thermal equilibrium with $B$). Curiously, it is followed by rapid oscillations between negative and positive entropy production rates. The timescale of these oscillations is determined by $g$ and the detuning $\Delta$ (more precisely the largest relevant $\Omega_n$, determined by the associated Boltzmann weight $e^{-n\omega_B\beta_B}/Z$). Such oscillations are anti-correlated with the rate of the coupling energy $\dot E_\text{int}={\rm Tr}[\dot\rho_{AB}(t)V_{AB}]$, with $V_{AB} = g(\sigma_+ a+ \sigma_- a^\dag)$  (see Fig.~\ref{figapp:covsEint} in Appendix \ref{app:addplots}). Additional information and interpretation of these oscillations are provided in Appendix \ref{app:addplots}.  

A third observation is the decay of the oscillations, as if an equilibration process was occurring. The latter happens on a timescale determined by $\gamma_1(t)$, $\gamma_2(t)$, $\gamma_3(t)$, and therefore by the coupling strength $g$. However, there is a revival of the oscillation, for $gt \gg 1$, because the overall dynamics is quasi-periodic.

In Fig.~\ref{fig:all}, we represent the time-evolution of the different definitions of entropy production in a stronger coupling regime, with $g/\omega_A = 0.1$ and $\Delta/\omega_A = 0.15$, while the other parameters have the same values as in Fig.~\ref{fig:all_weak}. We verified numerically that the RWA yielding the Jaynes-Cummings model is still valid for these values of $g$ and $\Delta$. One can see in Fig.~\ref{fig:all} many important differences from weak coupling (a zoom of Fig.~\ref{fig:all} is provided in Fig.~\ref{figapp:zoom} of  Appendix~\ref{app:addplots}). Firstly, there is a major entropy production peak around $gt \simeq 0.5$, a value roughly 4 times larger than the initial peak at weak coupling. 
Secondly, there are some important discrepancies between the different definitions of entropy production, even though there are some similarities in the overall behavior: all entropy definitions, including the correlations $\dot I_{A:B}(t)$, present a large initial peak, followed by one oscillation and a transient close to zero, reminding of an equilibration phenomenon, before ending on oscillation revivals. Furthermore,  a closer inspection reveals that the definitions $\sigma^{Es}$ and $\sigma^{fp}$ are exactly the same. This coincidence is indeed exact, and is shown analytically in Appendix~\ref{app:sigmaid}. This is a very surprising result since $\sigma^{Es}$ is a definition involving the bath, as expected in strong coupling thermodynamics, while there is no explicit mention of the bath in the definition of $\sigma^{fp}$. However, it should be noted that such an identity between $\sigma^{Es}$ and $\sigma^{fp}$ is so far proved only for the Jaynes-Cummings model and when the mode $B$ is initially in a thermal state (this statement is not valid for arbitrary initial state of $B$, see Appendix~\ref{subsec:identity}). 

Noticing that $\sigma^{El}$ can be re-expressed as $\sigma^{El}(t) = \dot I_{A:B}(t) + \frac{d}{dt}D\{w[H_B,\beta^\text{eff}(t)]|\rho_B(0)\}$, the close proximity between $\sigma^{El}$ and $\dot I_{A:B}$ in Fig.~\ref{fig:all} reveals that the thermal background of $B$ remains at all times very close to the initial thermal state of $B$. Similarly, $\sigma^{Es}$ can be re-expressed as, $\sigma^{Es}(t) = \dot I_{A:B}(t) + \frac{d}{dt}D[\rho_B(t)|\rho_B(0)]$, implying that, contrary to its thermal background, $\rho_B(t)$ changes significantly, and even slightly oscillates, compensating the small oscillations of $\dot I_{A:B}(t)$. Interestingly, having significant changes in $\rho_B(t)$ while a thermal background $w[H_B, \beta^\text{eff}(t)]$ remaining almost constant suggests that $B$ is gaining only non-thermal energy \cite{Elouard2023}. Note that due to the time dependence of $\tilde H_A(t)$,  $\sigma^{Co}$ cannot be expressed in the form of $\dot I_{A:B}(t)$ plus a rate of relative entropy, which makes the physical interpretation of the origin of the discrepancies with the other entropy productions more challenging.

\begin{figure}
\begin{center}
    \includegraphics[width=0.45\textwidth]{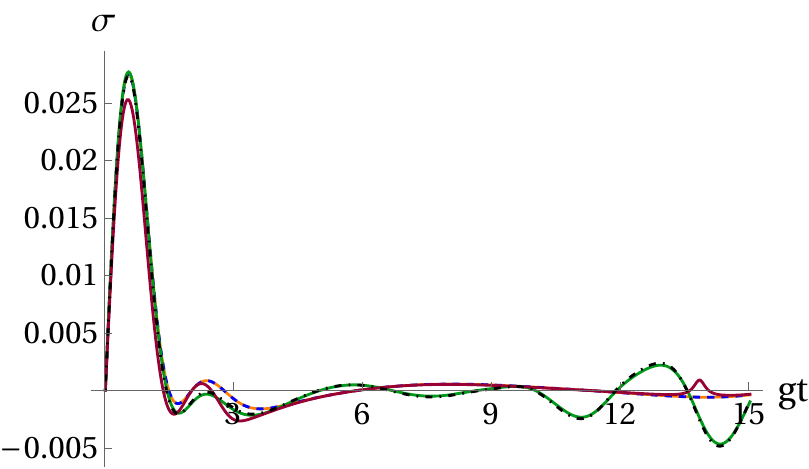}
\end{center}
\caption{Plots of the different definitions of entropy production $\sigma^{Es}$ (orange, thick), $\sigma^{El}$ (green), $\sigma^{Co}$ (purple),  $\sigma^{fp}$ (dashed blue), and $\dot I_{A:B}$ (dot-dashed black), in the strong coupling regime, with $\omega_B/\omega_A = 0.85$, $\Delta/\omega_A = 0.15$, $\omega_A\beta_A = 1.1$, $\omega_A\beta_B = 0.3$, and $g/\omega_A =0.1 $.}\label{fig:all}
\end{figure}

In addition to the above physical insights on entropy production in finite bath, our results show that at weak coupling, the comparison with memory effects can be made with any of the definitions of entropy production considered here. 
At strong coupling, there are two classes which are emerging from the above comparisons: the Esposito / fixed point entropy production, $\sigma^{Es}$ and $\sigma^{fp}$ (which coincide exactly), and the Elouard entropy production / correlations, $\sigma^{El}$ and $\dot I_{A:B}$, which are numerically close, but not exactly equal. Relations with memory effect are investigated in detail in Sec.~\ref{sec:entropymemory}.

\section{Entropy production and memory effects}\label{sec:entropymemory}
Memory effects in non-Markovian quantum systems can be characterized in many different ways, and constitute a vast field of research \cite{Rivas2014,Breuer2016, Vega2017}. The main criteria are CP-divisibility, P-divisibility, and the BLP criterion (see more details in Appendix~\ref{app:CPdivBLP} and~\cite{Rivas2014,Breuer2016, Vega2017, Theret2025}). 
For a two-level open system governed by a master equation of the form~\eqref{eq:ME}, P-divisibility can be characterized directly in terms of the  time-dependent damping rates $\gamma_1(t)$, $\gamma_2(t)$, and $\gamma_3(t)$, as follows: the dynamics is P-divisible at time $t$ if and only if~\cite{Filippov,Theret2025} 
\bea
& &|\gamma_-(t)| \leq \gamma_+(t)~\textrm{and}\nn\\
& &\text{if}~ 2\Gamma(t) \leq  \gamma_+(t), ~\text{then}~ \gamma_-(t)^2 \leq 4\Gamma(t)(\gamma_+(t) - \Gamma(t)),\nn\\  \label{Pcrit}
\eea
where $\gamma_\pm(t) := \gamma_1(t) \pm \gamma_2(t)$ and $\Gamma(t) = \gamma_3(t) +\gamma_2(t)/2$. 

\begin{figure}
\begin{center}
    (a)\includegraphics[width=0.45\textwidth]{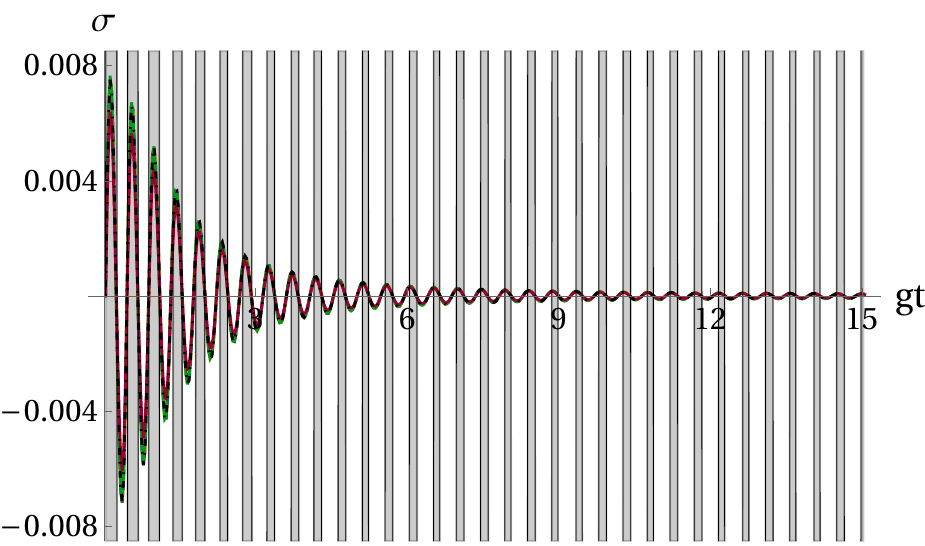}
\end{center}
\caption{Plots of the different definitions of entropy production $\sigma^{Es}$ (orange, thick), $\sigma^{El}$ (green), $\sigma^{Co}$ (purple),  $\sigma^{fp}$ (dashed blue), and $\dot I_{A:B}$ (dot-dashed black) in a weak coupling regime, with $\omega_B/\omega_A = 0.6$, $\Delta/\omega_A = 0.4$, $\omega_A\beta_A = 1.1$, $\omega_A\beta_B = 0.3$, and $g/\omega_A =0.03 $. The shaded grey area corresponds to the intervals of time where the dynamics is P-divisible (see the text for details).}\label{fig:Pdiv_weak}
\end{figure}

In Fig. \ref{fig:Pdiv_weak}, we use such a characterization of P-divisibility to compute the time intervals on which the dynamics is P-divisible, and compare it with the entropy production in the weak coupling regime. 
The correspondence between P-divisibility and positivity of entropy production is almost perfect. 

In Fig.~\ref{fig:Pdiv_minEP}(a), for  clarity, we focus on a shorter time interval and we plot only the entropy production $\sigma^{fp}$ (identical to $\sigma^{Es}$) for several initial states of $A$ (from pure states to maximally mixed states). One can see that the correspondence between P-divisibility is actually not exact and depends on the initial state. However, one can verify that whenever the dynamics is P-divisible, the entropy production for all plotted initial states is positive. This is actually a well-known consequence of the contractivity of the relative entropy under positive maps~\cite{Muller2017}: P-divisibility implies $\sigma^{fp} \geq0$  (see Appendix \ref{app:posep} for additional details). Still, the reverse is not true, as we can clearly see in Fig.~\ref{fig:Pdiv_minEP}(a).

\begin{figure}
\begin{center}
   (a) \includegraphics[width=0.45\textwidth]{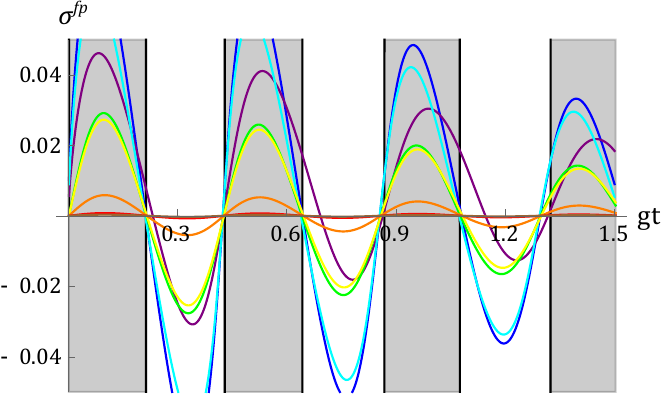}\\
    (b)\includegraphics[width=0.45\textwidth]{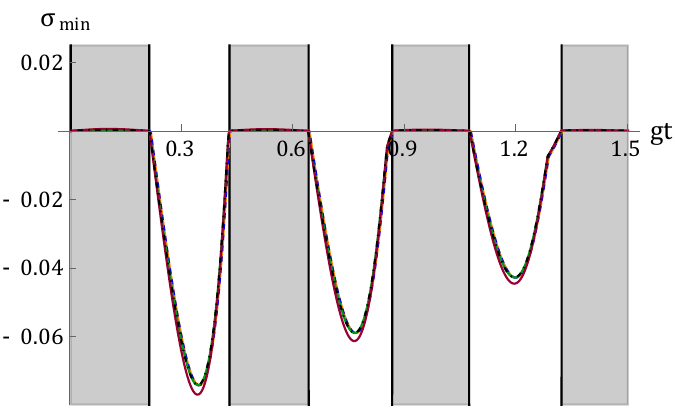}
\end{center}
\caption{(a) Plots of the entropy production $\sigma^{fp}$ for several initial states of $A$. (b) Plot of $\sigma^{fp}_\text{min}$ and of the minimum of all entropy production definitions (defined as in Eq. \eqref{sigmamin}). The chosen parameters correspond to the weak coupling regime, $\omega_B/\omega_A = 0.6$, $\Delta/\omega_A = 0.4$,  $\omega_A\beta_B = 0.3$, and $g/\omega_A =0.03 $. The shaded grey areas correspond to the intervals of time where the dynamics is P-divisible. }\label{fig:Pdiv_minEP}
\end{figure}

Nevertheless, it seems that whenever the dynamics is not P-divisible at time $t$, there exists at least one initial state such that the associated entropy production is indeed negative. This leads us to introduce the minimum of the entropy production over all initial states, 
\be\label{sigmamin}
\sigma_\text{min}^{fp}(t) := \min_{ \rho_A(0)} \sigma^{fp}(t).
\ee
One can indeed verify in Fig.~\ref{fig:Pdiv_minEP}(b) that there is a perfect correspondence between P-divisibility at time t and the positivity of the minimum of the entropy production $\sigma_\text{min}^{fp} $. Moreover, it appears also in Fig.~\ref{fig:Pdiv_minEP}(b) that such a correspondence holds for the minimum of all definitions of entropy production, defined in a similar way as $\sigma_\text{min}^{fp}(t)$ in Eq.~\eqref{sigmamin}. This correspondence  is expected because all entropy production definitions are equivalent at weak coupling (see Fig.~\ref{fig:all_weak}). Additionally, in Appendix~\ref{app:CPdivBLP}, we also provide the comparison with CP-divisibility and BLP criterion, revealing also a very close correspondence with the minimum of entropy production. This is because such criteria are almost equivalent in this case, as described also in~\cite{Theret2025}.

\begin{figure}
\begin{center}
   (a) \includegraphics[width=0.45\textwidth]{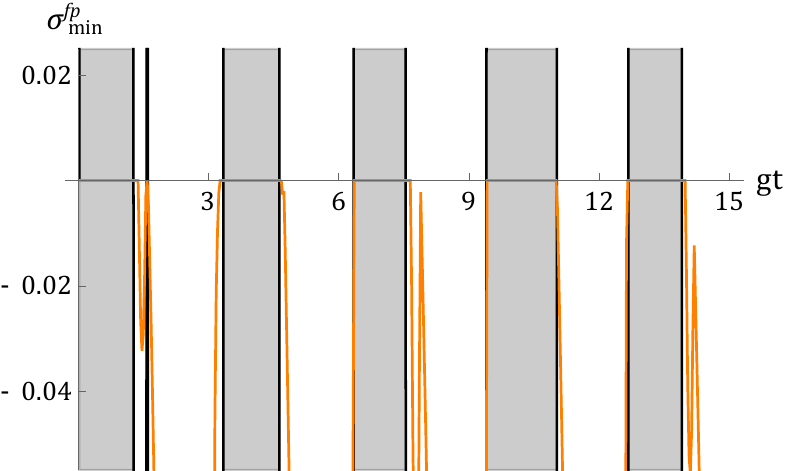}\\
    (b)\includegraphics[width=0.45\textwidth]{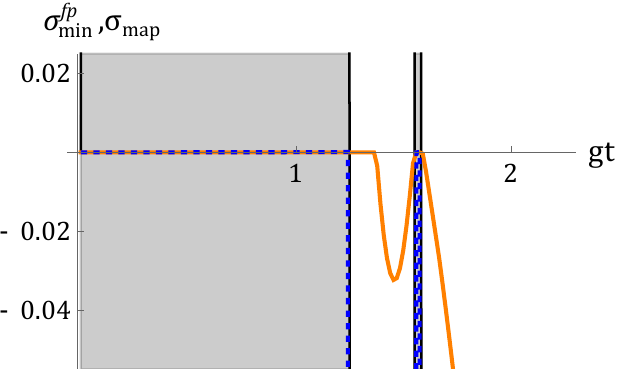}
\end{center}
\caption{(a) Plot of $\sigma^{fp}_\text{min}$, minimum of the entropy production over all initial states, in the strong coupling regime $g/\omega_A =0.3 $, near to resonance, $\omega_B/\omega_A = 0.99$, $\Delta/\omega_A = 0.01$, and cold bath $\omega_A\beta_B = 3$. (b) Zoom of plot (a) on a region where the equivalence between P-divisibility and the sign of $\sigma^{fp}_\text{min}$ fails. The blue dashed line corresponds to the plot of $\sigma_\text{map}$, the map entropy production defined in the main text in Eq.~\eqref{eq:sigmamap}. Note that the minimum over all states in the definition of $\sigma_\text{map}$ is obtained numerically by a discrete parameterization of the Bloch ball. However, for some instant of times, this minimum cannot be accessed numerically. Instead we perform an analytical analysis to obtain it. See Appendix \ref{app:findingsigmamap} for more details.}\label{fig:Pdiv_minEP_strong}
\end{figure}

However, such a correspondence between P-divisibility and the sign of $\sigma^{fp}_\text{min}(t)$ is lost for stronger coupling. This is illustrated in Fig.~\ref{fig:Pdiv_minEP_strong}, where there exist times (around $gt=1.35$, $gt=7.65$ and $gt=14$) for which the dynamics is not P-divisible, and still $\sigma^{fp}_\text{min}(t) > 0$. It means that the relation between P-divisibility and the entropy production $\sigma^{fp}_\text{min}$ is more subtle. Note that in Fig.~\ref{fig:Pdiv_minEP_strong}, the chosen coupling strength breaks the RWA, but this was made merely for illustrative purposes to present significant discrepancies between the sign of $\sigma^{fp}_\text{min}(t)$ and P-divisibility. 
From the P-divisibility criterion~\eqref{Pcrit} derived in~\cite{Theret2025}, it appears that P-divisibility is a property which depends only on the instantaneous generator ${\cal L}_t$, and not on the integrated dynamics from $0$ to $t$. However, $\sigma^{fp}_\text{min}(t)$ does depend on the integrated dynamics from $0$ to $t$, that we denote $\Lambda_{t,0}$,
\bea
&&\sigma^{fp}_\text{min}(t)= \nn\\
&& \min_{ \rho_A(0)}  \Big\{-  {\rm Tr} \left\{{\cal L}_t \Lambda_{t,0}\rho_A(0)\left[ \ln \Lambda_{t,0} \rho_A(0) - \ln \rho_A^{fp}(t)\right] \right\}\Big\},\nn
\eea
 This observation suggests to introduce the \emph{map entropy production}, as an extension of the entropy production to the map, as 
\bea\label{eq:sigmamap}
\sigma_\text{map}(t) &:=& \min_{ \rho_A}  \Big\{-  {\rm Tr} \left\{{\cal L}_t \rho_A\left[ \ln \rho_A - \ln \rho_A^{fp}(t)\right] \right\}\Big\} \nn\\
&&\leq \sigma^{fp}_\text{min}(t).
\eea
Note that $\sigma_\text{map}$ depends only on the instantaneous generator ${\cal L}_t$ and is a natural lower bound of $\sigma_\text{min}^{fp} $.
We see in Fig.~\ref{fig:Pdiv_minEP_strong}(b) a perfect correspondence between the sign of the map entropy production $\sigma^{fp}_\text{map}$ and P-divisibility. To confirm this numerical result, we demonstrate analytically such equivalence in the case of phase-covariant dynamics, and present the main points of the proof in Sec.~\ref{sec:proof}. 
This highlights an important link between, first, irreversibility and lack of memory effect, and secondly, between negative entropy production rate and memory effects.

\section{Equivalence between P-divisibility and non-negativity of $\sigma_\text{map}$}\label{sec:proof}

Based on the numerical results performed in the case of a qubit interacting with a single bosonic mode, the goal of this section is to prove rigorously the equivalence between P-divisibility and non-negativity of $\sigma_\text{map}$. However, in the proof, we consider not only this example, but also a more general two-level open quantum system with phase-covariant dynamics~\cite{Filippov,Theret2025}, where we neglect all the nonsecular terms of the master equation. In other words, it means that we consider a master equation of the form~\eqref{eq:ME} with arbitrary time-dependent coefficients $\Omega_A(t)$ and  $\gamma_i(t)$, $i=1,2, 3$. For convenience, we introduce the notations $\gamma_\pm(t) := \gamma_1(t) \pm \gamma_2(t)$, and $\Gamma(t) := \gamma_3(t) + \gamma_+(t)/2$.

Let $(x(t),y(t),z(t))=(\textrm{Tr}[\rho_A\sigma_x],\textrm{Tr}[\rho_A\sigma_y],\textrm{Tr}[\rho_A\sigma_z])$ be the Bloch coordinates of the density operator $\rho_A(t)$ of the system. We introduce the modulus $r$ of the Bloch vector with $r^2(t) = x^2(t) + y^2(t) + z^2(t)$. In a specific rotating frame, the dynamics of the Bloch coordinates are governed by the following set of differential equations~\cite{Theret2025,Theret2023}
\begin{eqnarray*}
& &\dot{x}=-\Gamma(t) x, \\
& &\dot{y}=-\Gamma(t) y,\\
& &\dot{z}=\gamma_-(t)-\gamma_+(t)z.
\end{eqnarray*}
Starting from Eq.~\eqref{sigmafp}, straightforward computations show that, for all $t$ such that $r(t)\neq1$,
\begin{equation}
\sigma^{fp}(t) = -\dot r(t)L(r(t)) + \dot z(t)L(\gamma_\infty(t)),
\end{equation}
where $L(x) = \textrm{arctanh}(x) =  \frac{1}{2}\ln\!\left(\frac{1 + x}{1 - x}\right)$ and $\gamma_\infty(t) = \frac{\gamma_-(t)}{\gamma_+(t)}$.
Note that the entropy might be infinite for pure states $r=1$. Infinite values can be avoided when $\dot r_{|r = 1} = 0$.

From now on we consider only the case where $\gamma_\infty(t)$ is constant, which is assumed to be positive, 
\begin{equation}
\gamma_\infty(t) = z_\infty \geq 0.
\end{equation}
The case $\gamma_\infty(t)$ a negative constant yields the same results. We then get, 
\begin{equation}
\label{eq:1}
\sigma^{fp} = \Gamma \ell^2 f(r) + \gamma_+(z_\infty - z)\left(z_\infty f(z_\infty) - zf(r)\right),
\end{equation}
where $f(r) \equiv \frac{1}{r}L(r) = \frac{1}{2r}\ln\!\left(\frac{1+r}{1-r}\right)$ is strictly increasing from 1 to $+\infty$, and  $\ell(t)^2 = x(t)^2 + y(t)^2$. As expected we see that $\sigma^{fp}$ is invariant under rotations around the $z$-axis and therefore only depends upon $z$ and $\ell = \sqrt{x^2 + y^2}$.

We now consider $\sigma^{fp}$ as a function of the parameters $(\ell,z)$ defined within the Bloch ball $\frak B = \{(\ell,z)\ :\ \ell^2 + z^2 \leq 1\}$. This differs from considering the entropy production obtained by allowing the evolution to run for a time $t$, where only the points $(\ell(t),z(t))$ attained by trajectories at time $t$ would be taken into account.
According to Eq. \eqref{eq:sigmamap}, the map entropy production can be defined as  
$$
\sigma_{\rm map}(t) = \inf_{(\ell,z) \in\frak{B}}\sigma^{fp}(t).
$$
We shall prove the following result (see Appendix~\ref{appA}).

\begin{theorem}\label{theo_entropymap}
For a two-level system undergoing a phase-covariant dynamics of the form ~\eqref{eq:ME}, for any $t\geq0$, $\sigma_{\rm map}(t) \geq 0$ if and only if the dynamics is P-divisible at time $t$.
\end{theorem}

It is already known in full generality that P-divisibility at time $t$ implies $\sigma^{fp}(t)\geq0$ for all initial states (see Appendix \ref{app:posep} and ~\cite{Muller2017}). However, the reverse has not been proven. Here we present a direct proof for a class of two-level quantum systems, which has the advantage of identifying trajectories where the non-negativity of $\sigma^{fp}$ and therefore P-divisibility might fail to occur. 

The structure of the proof is as follows. The mathematical details are given in Appendix~\ref{appA}. We first argue that in order for $\sigma^{fp}$ to be non-negative we must have $\Gamma\geq0$ and $\gamma_+\geq0$, which is also a necessary requirement for P-divisibility. Under these conditions we show that the positivity of $\sigma^{fp}$ might fail only inside a subset of $\mathfrak B$ called \emph{the critical zone} $\mathcal{C}$.
The proof then goes on by studying the sign of $\sigma^{fp}$ along the horizontal slices $z=\rm cste$ of $\mathfrak B$. It turns out that the sign of $\sigma^{fp}$ along these slices is given by the sign of $\dot r_{|r=1}$, which is directly related to P-divisibility~(see~\cite{Theret2025}). 
This eventually proves that $\sigma_\text{map}\geq0\Rightarrow$ P-divisibility.
The converse is obtained by directly computing the sign of $\sigma$ for a P-divisible dynamical map at time $t$.

\section{Conclusion}\label{sec:conclusion}
In this work, we have established a precise link between two key phenomena in open quantum dynamics, i.e. memory effects and irreversibility. 
Irreversibility  is quantified by entropy production, which is well understood for systems weakly coupled to an infinite bath, but still present a lack of consensus for strong coupling and finite bath situations. The first part of the paper was therefore dedicated to the comparison between several suggestions of entropy production definitions. The model we consider is a challenging situation combining a minimal bath (a single bosonic mode) with strong coupling. 
We have shown that all definitions of entropy production considered in this paper coincide at weak coupling, highlighting their relevance, while important discrepancies appear at stronger coupling. 
We have also analytically shown that two definitions coincide exactly (see the proof in Appendix~\ref{app:sigmaid}): the widely used definition  $\sigma^{Es}$ based on the energy variation of the bath~\cite{Esposito2010}, and a definition based on the instantaneous fixed point of the map, $\sigma^{fp}$. This is a surprising coincidence because, conversely to $\sigma^{Es}$, $\sigma^{fp}$ depends only on local quantities of the system. It should be noted that such an identity was obtained for the Jaynes-Cummings model, with the bosonic mode initially in a thermal state. Whether this conclusion is also valid for larger systems as long as the RWA is applied and the system $B$ is initially in a thermal state is an open question. Note that for the specific case of pure dephasing it was recently shown in ~\cite{Picatoste2025} that the identity is not valid.  

Regarding memory effects, many measures exist, including CP-divisibility, P-divisibility, and the BLP criterion. At weak coupling, they mostly coincide (see Appendix \ref{app:CPdivBLP} and \cite{Theret2025}), and we have found a very good correspondence with entropy production. More precisely, the dynamics is P-divisible at time $t$ if and only if entropy production is positive for all initial states.

However, as expected, the three figures of merit exhibit significant differences at strong coupling. We find that only P-divisibility shows a close correspondence with the entropy production based on the fixed point $\sigma^{fp}$. Still, the correspondance is not as perfect as for weak coupling: there exist small time intervals on which the dynamics is not P-divisibible but the entropy production is positive for all initial states. A perfect equivalence is obtained by introducing $\sigma_\text{map}$, a notion of entropy production at the map level, inspired from the local character of P-divisibility. We show that the dynamics is P-divisible whenever $\sigma_\text{map}$ is positive. This numerical observation is also confirmed by a rigorous proof valid for a large class of phase-covariant dynamics. Moreover, a very recent work \cite{Picatoste2026} extends our result to any CPTP map (quantum map), including maps which are not phase-covariant, and to any finite-dimensional quantum systems.

More than indicating that P-divisibility and $\sigma^{fp}$ are respectively the most relevant figures of merit for memory effects and irreversibility, this result highlights an important link between irreversibility and lack of memory effects. Conversely, it provides an interpretation of negative entropy production rates in terms of memory effects.

Finally, our results open many perspectives, such as a better understanding of what happens when the steady state is not a physical state (which can occur when the bath is not initially in a thermal state), and how to precisely characterise weak and strong coupling in finite bath situations. It also invites to further studies on the relation between memory effects and irreversibility in open quantum dynamics combining the viewpoints of open quantum systems and quantum thermodynamics, specially in settings where the memory effects can be directly quantified by the size of an explicit memory assisting a Markovian thermal processes \cite{Czartowski2023}, as well as in the light of the difference between revivals of information versus backflows of information, pointed out recently \cite{Buscemi2025}. Additionally, a natural next step is to extend this study to quantum thermodynamic frameworks based on the mean force Gibbs state \cite{Rivas2020Apr,Gonzalez2025,Strasberg2019Oct, Strasberg2020May}. It should also be pertinent to study these issues for controlled quantum systems in the context of open dynamics~\cite{koch2016,kochroadmap,lapert2013,mukherjee2013,tutorial24}.

\acknowledgments

C.L.L. acknowledges funding from the French National
Research Agency (ANR) under grant ANR-23-CPJ1-
0030-01.\\

\appendix

\section{Proof of Theorem~\ref{theo_entropymap}}\label{appA}

Notation: in this appendix, the fixed point entropy production is denoted for simplicity by $\sigma(t)\equiv \sigma^{fp}(t)$.

\subsection{Specific trajectories and the critical zone}
First, we study specific trajectories in the Bloch ball $\frak B$. It is easily checked that the trajectories in the plane $z = z_\infty$ at a given time $t$ remain in this plane for all times. Hence, the sign of $\sigma$ is the same as that of $\Gamma$ (see Eq. (\ref{eq:1})). Likewise, the trajectories along the $z$- axis ($\ell = 0$) stay on this axis and $\sigma = \gamma_+ K$ with $K\geq0$. Hence, the sign of $\sigma$ is the same as the sign of $\gamma_+$.

We deduce that if the system is not P-divisible because one of the coefficient rates becomes negative, then these trajectories result in negative entropy production. Therefore, $\sigma_{\rm map}(t) \geq0 \Rightarrow \gamma_+(t)\geq0, \Gamma(t)\geq0$.

\begin{lemma}
Assume that $\gamma_+(t)\geq0$ and $\Gamma(t)\geq0$.
The entropy production $\sigma(t)$ might fail to be non negative only if $0\leq z(t) < z_\infty$ and $r(t)\geq z_\infty$.
\end{lemma}
The set 
$$
\mathcal C = \{(\ell,z)\in\frak B\ :\ 0\leq z < z_\infty\ \textrm{and}\ r\geq z_\infty\}
$$
is called the \emph{critical zone} as shown in Fig.~\ref{fig:critzone}.
\begin{figure}
\begin{center}
\includegraphics[width=0.45\textwidth]{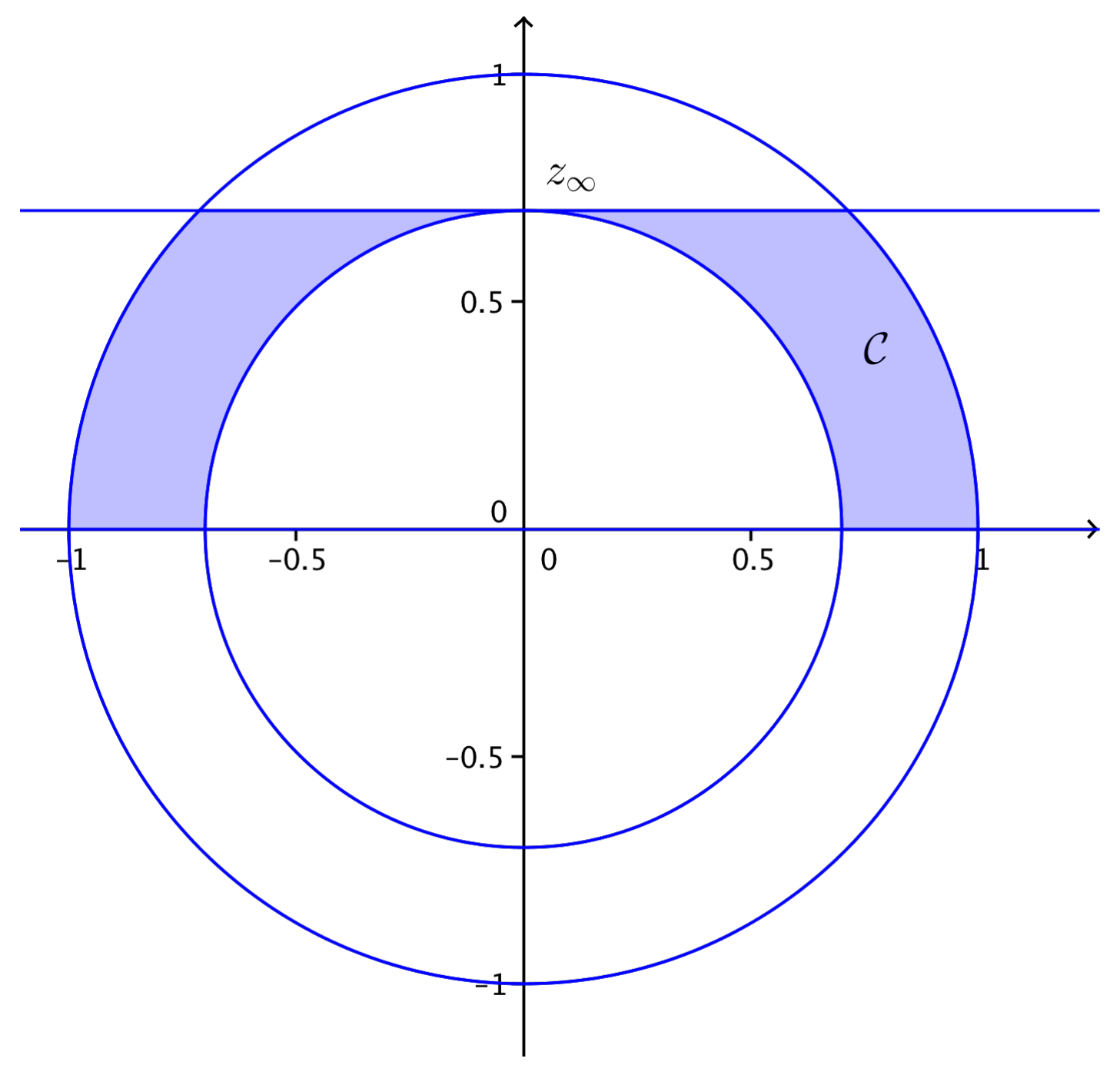}
\end{center}
\caption{Schematic illustration of the critical zone $\mathcal{C}$ (in blue) in the $(y,z)$- plane.}\label{fig:critzone}
\end{figure}

\begin{proof}
From Eq.~(\ref{eq:1}) we can write $\sigma$ as
\begin{equation}\label{eq:A1}
\sigma(t) = 2\Gamma A + \gamma_+(B + C),
\end{equation}
where
$$
A = \frac{\ell^2}{2}f(r),~B = (z - z_\infty)zf(r),
~C = -(z - z_\infty)z_\infty f(r_\infty).
$$
Recall that, as mentioned in the main text, $\sigma(t)$ is the fixed point entropy production extended to the entire Bloch ball. This means that in Eq.~\eqref{eq:A1}, differently from Eq.~\eqref{eq:1}, the parameters $(\ell,z)$ are arbitrary inside the Bloch ball $\frak B = \{(\ell,z)\ :\ \ell^2 + z^2 \leq 1\}$, which is different from considering the entropy production obtained by allowing the evolution to run for a time $t$, where only the points $(\ell(t),z(t))$ attained by trajectories at time $t$ would be taken into account.

In Eq.\eqref{eq:A1} we have $A\geq0$. The study can be divided into the following cases:.
\begin{enumerate}
\item $z - z_\infty \geq 0$. Then $r \geq z \geq z_\infty \geq 0$, $B\geq 0$, $C \leq 0$ and $B+C\geq0$.
\item $z - z_\infty \leq 0$. Then $C \geq 0$ and $z\leq z_\infty$.
\begin{enumerate}
\item $z\leq0$. Then $B \geq 0$ and $B + C \geq 0$.
\item $z\geq0$. Then $B \leq 0$.
\begin{enumerate}
\item $r\leq z_\infty$. Then $B + C \geq 0$.
\item $r\geq z_\infty$. This writes $\ell^2 + z^2 - z_\infty^2 \geq 0$, i.e., $\ell^2 + (z - z_\infty)(z + z_\infty) \geq 0$.
Hence, $\ell^2 \geq (z_\infty - z)(z + z_\infty) \geq 0$.
This gives $A+B \geq \frac{f(r)}{2}\left(z_\infty - z\right)^2\geq0$.
Therefore, $A + B \geq 0$.
\end{enumerate}
\end{enumerate}
\end{enumerate}
Examining the different cases shows that, whenever $(\ell,z)$ is outside the critical zone, we have that $0\leq\sigma$.\qed
\end{proof}
Notice that $A + B + C \geq 0$ in any case. In particular, if the evolution is CP-divisible at time $t$, we have $2\Gamma\geq\gamma_+$.
Then $\sigma_{\rm map}\geq\gamma_+(A+B+C)\geq0$.

\subsection{Study of the sign of $\dot r$}

\noindent
Since the time $t$ is fixed, we can drop it. Straightforward computations give 
\begin{eqnarray*}
\dot \ell &=& -\Gamma\ell,\\
\dot z &=& \gamma_+(z_\infty - z),\\
\dot r &=& \frac{1}{r}\left(-\Gamma\ell^2 + \dot z z\right),\\
\sigma &=& -\dot r L(r) + \dot zL(z_\infty).
\end{eqnarray*}
We study how these functions behave with respect to the parameter $\ell$, which means that we fix the coordinate $z$ and examine the functions over horizontal slices of the Bloch ball.
To remain within the Bloch ball, the parameter $\ell$ must stay within the interval $0 \leq \ell \leq \sqrt{1 - z^2}$, but $\ell$ belongs to $[0,+\infty[$ in general.

\begin{lemma}
Suppose that $\gamma_+,\Gamma\geq0$ and that $0\leq z \leq z_\infty$.
There exists a unique $\hat\ell = \hat\ell(z) \geq0$ such that the function $\ell\mapsto\dot r(\ell)$ is positive and strictly decreasing over $\ell\in[0,\hat\ell]$, $\dot r(\hat\ell) = 0$, and the function $\ell\mapsto\dot r(\ell)$ is negative and strictly decreasing over $\ell\in[\hat\ell,+\infty]$.
\end{lemma}

\begin{proof}
By differentiating with respect to $\ell$ we get that $\frac{d}{\ell}\dot r(\ell)$ has the same sign as $-\Gamma(\ell^2 + 2z\dot z) - z\dot z$.

Since $\gamma_+\geq0$ and $0\leq z \leq z_\infty$, we have $\dot zz \geq 0$.
Since $\Gamma\geq0$, we deduce that the function $\ell \mapsto \dot r(\ell)$ is strictly decreasing over $[0,+\infty[$, from $\dot r(0) = \dot z \geq 0$ to $\lim_{\ell\to+\infty}\dot r(\ell) = -\infty$ (unless $\Gamma(t) = 0$, in which case the function is constant). We thus conclude that there exists a unique $\hat\ell$ such that $\dot r(\hat \ell) = 0$.
Solving for $\dot r = 0$, we get  
$$
\hat\ell = \sqrt{\frac{\dot z z}{\Gamma}}.
$$
On the interval $[0,\hat\ell]$ over which $\dot r(\ell)\geq0$, the map $\ell \mapsto -\Gamma\ell^2 + \dot z z$ is non-negative and it is also strictly decreasing.
Since $\ell\mapsto\frac{1}{r(\ell)}$ is also non-negative and strictly decreasing, we conclude that the map $\ell \mapsto \dot r(\ell)$ is strictly decreasing over $[0,\hat\ell]$.$\hfill\square$ 
\end{proof}

From the proof we have $\ell(z) = \sqrt{\frac{\dot z z}{\Gamma}}$ if $\Gamma\neq0$.
Let us determine the conditions under which $(\hat\ell(z),z)$ belongs to the critical zone $\mathcal C$.
\begin{lemma}
Suppose that $\gamma_+,\Gamma\geq0$ and that $0\leq z \leq z_\infty$.
Set $\hat r(z)^2 \equiv \hat\ell(z)^2 + z^2$.
Then $\hat r(z) \geq z_\infty$ if and only if
\begin{equation}
\label{Condition_ell1}
2\Gamma < \gamma_+\ \textrm{and}\ z\in[z_+,z_\infty],
\end{equation}
where $z_+ = z_\infty\frac{\Gamma}{\gamma_+ - \Gamma}$.
In particular, if $2\Gamma \leq \gamma_+$, then $\hat r(z) \geq z_\infty$ for all $z\in[0,z_\infty]$.  
\end{lemma}
\begin{proof}
For $0\leq z \leq z_\infty$, the inequality $\hat r(z) \geq z_\infty$ is equivalent to  
$$
-\Gamma z_\infty^2 + \gamma_+ z_\infty z + (\Gamma - \gamma_+)z^2 \geq 0.
$$
If $\gamma_+ = \Gamma$, then the inequality holds if and only if $z\geq z_\infty$.
If $\gamma_+\neq \Gamma$, we get a second-degree polynomial in $z$ whose sign is easy to determine. The discriminant of the polynomial is 
$$
\Delta = z_\infty^2(2\Gamma - \gamma_+)^2 \geq 0.
$$
The roots are thus
\begin{equation}
z_\pm = z_\infty\frac{-\gamma_+ \pm (2\Gamma - \gamma_+)}{2(\Gamma - \gamma_+)}
=
\left\{
\begin{array}{ll}
z_- &= z_\infty\\
z_+ &= z_\infty\frac{\Gamma}{\gamma_+ - \Gamma}
\end{array}
\right..
\end{equation}
Suppose that $\Gamma>\gamma_+$, then $z_+ <0$ and the polynomial is negative for all $z\in[0,z_\infty]$. Suppose that $\Gamma<\gamma_+$, then $z_+ > z_- \iff 2\Gamma > \gamma_+$. We finally obtain the following statements.
If $2\Gamma > \gamma_+$ then the polynomial is negative for all $z\in[0,z_\infty]$. If $2\Gamma < \gamma_+$ then the polynomial is negative for all $z\in[0,z_+]$, then positive for all $z\in[z_+,z_\infty]$. The interval is reduced to $\{z_\infty\}$ if $2\Gamma = \gamma_+$.\qed
\end{proof}
We have to consider now the condition $\hat\ell\leq\ell_1=\sqrt{1 - z^2}$.
This condition easily implies that $\dot r_{|r = 1}\leq0$. It has been shown in~\cite{Theret2025} that the condition $\dot r_{|r = 1}\leq0$ for all $r$ is equivalent to P-divisibility.
\subsection{P-divisibility is necessary}
Suppose that the evolution is not P-divisible.
If this occurs due to one of the coherence rates being negative, then we can conclude that $\sigma_{\rm map} < 0$. Suppose now that both coherence rates are non-negative. Since the evolution is not P-divisible, there exists $z$ such that $\dot r_{|r = 1}>0$.
Examining the formula for $\sigma$, 
we get, for this value of $z$,  $\lim_{\ell\to\ell_1}\sigma = -\infty$. We conclude that $\sigma_{\rm map} < 0$.
We thus have shown that $$\sigma_{\rm map} \geq 0 \Rightarrow \textrm{P-divisibility}.$$
\subsection{P-divisibility is sufficient}
Now suppose that the evolution is P-divisible at time $t$.
Since $\gamma_+(t)\geq0$ and $\Gamma(t) \geq 0$, we know that $\sigma\geq0$ everywhere outside the critical zone $\mathcal C$. We therefore consider the situation where 
\begin{equation}
2\Gamma < \gamma_+\ \textrm{and}\ z\in[0,z_\infty].
\end{equation}
In the P-divisibility case, we can also assume that $2\Gamma \geq \alpha\gamma_+$, where $\alpha = 1 - \sqrt{1 - z_\infty^2}$ (recall that $\gamma_- = z_\infty \gamma_+$). 
Recall that, since we are in the critical zone, $z\in[0,z_\infty]$ and for each such $z$, $\ell_\infty = \sqrt{z_\infty^2 - z} \leq \ell \leq 1$. For each point $(\ell,z)$, the value of $-\dot r$ increases with $\Gamma$, so does $\sigma$.
We conclude that the values of $\sigma$ are bounded from below by those of $\sigma_m$ obtained by taking $\Gamma = \frac{\alpha}{2}\gamma_+$.
Hence we restrict to this case, and $\sigma_m$ has the same sign as 
\begin{eqnarray*}
s_m &=& \frac{\alpha\ell^2}{2r}L(r) + (z_\infty - z)\left(L(z_\infty) - \frac{z}{r}L(r)\right)\\
&=& \alpha\frac{r^2 - z^2}{2r}L(r) + (z_\infty - z)\left(L(z_\infty) - \frac{z}{r}L(r)\right)\\
&=& \frac{L(r)}{r}\left(\left(1-\frac{\alpha}{2}\right)z^2 - z_\infty z + \frac{\alpha}{2} \right) \\
& & + \frac{L(r)}{r}\frac{\alpha}{2}(r^2 - 1) + (z_\infty -z)L(z_\infty).
\end{eqnarray*}
It is easily checked that the coefficient $\left(1-\frac{\alpha}{2}\right)z^2 - z_\infty z + \frac{\alpha}{2}$ is non-negative for all $z$. Moreover, the function $\ell\mapsto\frac{L(r)}{r}\left(\left(1-\frac{\alpha}{2}\right)z^2 - z_\infty z + \frac{\alpha}{2} \right)$ is strictly increasing.
We then consider the second term
$$
A = \frac{L(r)}{r}\frac{\alpha}{2}(r^2 - 1) + (z_\infty -z)L(z_\infty).
$$
As a function of $\ell$ we get
$$
\frac{d A}{d\ell} = \frac{\alpha}{2}\frac{dr}{d\ell}\frac{d}{dr}\left(\frac{L(r)}{r}(r^2 - 1)\right).
$$
We find that $\frac{d A}{d\ell}$ has the sign of $B = L(r) - \frac{r}{1+r^2}$.
We have 
$$
\frac{d B}{dr} = \frac{4r^2}{(1-r^4)(1 + r^2)}\geq0.
$$
We conclude that the function $B(r)$ is strictly increasing from $B(0) = 0$, hence the function $\ell\mapsto A(\ell)$ is strictly increasing.
We deduce that the function $\ell\mapsto s_m(\ell)$ is also strictly increasing.
Moreover, we have
$$
s_m(\ell_\infty) = \frac{L(z_\infty)}{z_\infty}\left(\frac{\alpha}{2}(z_\infty^2 - z^2)+(z_\infty - z)^2\right)\geq0.
$$
Hence, for all $\ell\geq\ell_\infty$, $s_m(\ell)\geq0$, we conclude that $\sigma_{\rm map}\geq0$. The condition of P-divisibility is therefore sufficient.

\section{Finding the sign of $\sigma_\text{map}$}\label{app:findingsigmamap}

In this section, we provide a brief overview of the analytical analysis used to determine the sign of the minimum in the definition of map entropy production,  $\sigma_\text{map}(t) = \min_{ \rho_A}  \Big\{\tilde\sigma_{fp}[\rho_A]\Big\}$, with $\tilde\sigma_{fp}[\rho_A]:= -  {\rm Tr} \left\{{\cal L}_t \rho_A\left[ \ln \rho_A - \ln \rho_A^{fp}(t)\right] \right\}$, the fixed point entropy production extended to the whole Bloch ball.
For a state $\rho_A$ with Bloch coordinates $x$, $y$, and $z$, we denote by $\theta$ and $\phi$ the usual spherical angles, and we recall the notation for the norm of the Bloch vector, $r:=\sqrt{x^2+y^2+z^2}$. 
With these parameters, $l^2 = x^2+y^2 = r^2 \sin^2\theta$ and Eq.~\eqref{eq:A1} becomes
\bea
\tilde\sigma_{fp}[\rho_A]&=& \Gamma(t)  r L(r) \sin^2\theta \nn\\
 &&- \gamma_{+}(t)(r\cos\theta - z_\infty)[L(z_\infty)- \cos\theta L(r)].\nn\\
\eea
This can be recast as a function of $X:=\cos\theta$,
\bea
g(X):=\tilde\sigma_{fp}[\rho_A] &=& \Gamma  r L(r)(1-X^2) \nn\\
 &&- \gamma_{+}(rX - z_\infty)[L(z_\infty)- X L(r)].\nn\\
\eea
The function $g$ is a simple polynomial function of $X$. If $\gamma_+-\Gamma > 0$ ($\gamma_+-\Gamma < 0$), $g$ is convex (concave), with a minimum (maximum) at 
\be
X_0 := \frac{\gamma_+[rL(z_\infty) +z_\infty L(r)]}{2 r L(r)(\gamma_+-\Gamma)}.
\ee
If $\gamma_+ >0$, then $X_0 <0$, while $X_0 >0$ if $\gamma_+ <0$.
In particular, for the region on which the numerical minimization is not tractable (around the time $gt = 1.26$) in the plot of Fig.~\ref{fig:Pdiv_minEP_strong}(b) of the main text, $\gamma_+(t)$ is positive. Since $X:=\cos(\theta)$, the acceptable values of $X$ are therefore between -1 and 1. Then, for values of $r$ such that $X_0 \leq -1$, the minimum of $\tilde\sigma_{fp}[\rho_A]$ is reached for $X = -1$. Substituting in $g$, we obtain $g(-1) = \gamma_+(r+z_\infty)[L(z_\infty)+L(r)])$, which is always positive. 

However, for values of $r$ such that $X_0\geq -1$, the minimum of $\tilde\sigma_{fp}[\rho_A]$ is reached for $X = X_0$, which leads to 
\bea
g(X_0) &=& -\frac{1}{4}\frac{\gamma_+^2[rL(z_\infty)+z_\infty L(r)]^2}{rL(r)(\gamma_+-\Gamma)} + \Gamma r L(r)\nn\\
&&+ \gamma_+z_\infty L(z_\infty).
\eea
For $r \rightarrow 1$, $L(r) \rightarrow \infty$, so that
\be
g(X_0) \underset{r\rightarrow 1}{\sim} \left(\Gamma -\frac{\gamma_+^2z_\infty^2}{4(\gamma_+-\Gamma)}\right) L(r) < 0.
\ee
Additionally, one can verify that for the parameter settings of Fig.~\ref{fig:Pdiv_minEP_strong}(b), $X_0 \geq-1$ for $r \rightarrow 1$. We conclude that the minimum of $\tilde\sigma_{fp}[\rho_A]$ over all state is negative, so that the map entropy production $\sigma_\text{map}$ is negative at $gt= 1.26$. The same analytical reasoning can be reproduced for any instant of time.

\section{Additional results}\label{app:addplots}
In this section we present some additional results in Figs. \ref{figapp:covsEint}, \ref{figapp:popvsEint}, \ref{figapp:zoom}, and \ref{fig:RWA}, to illustrate some points mentioned in Sec.~\ref{sec:EPcomparison} on the comparison between the different entropy productions as well as the justification of the RWA in Fig.~\ref{fig:all}.

Figure~\ref{figapp:covsEint} presents the plots of the correlation build-up rate $\dot I_{A:B}$ and the time-derivative of interaction energy $\dot E_\text{int}(t) = {\rm Tr}[\dot \rho_{AB}(t) V_{AB}]$, in a weak coupling regime. We observe a clear anti-correlation between the oscillations of the two quantities. Therefore, a decrease in the coupling energy is accompanied by an increase in correlation and a loss of local information, expressed as a positive entropy production. Note that the integrated interaction energy is always negative, as is typical in correlated systems (see Appendix A of~\cite{Latune2023Aug}), so that a decrease of interaction energy results in an increase in absolute value. Conversely, an increase in coupling energy is followed by a decrease in correlations and an increase in local information.

Figure~\ref{figapp:popvsEint} shows that the oscillations of the coupling energy $\dot E_\text{int}(t)$ are related to the oscillations of $\dot p_A(t) = \bra1 |\rho_A(t)|1\ket$, which themselves follow the exchanges of quanta of energy between $A$ and $B$. This makes sense because during an exchange of quanta, the imbalance between $\omega_A$ and $\omega_B$ is compensated by a contribution from the interaction energy. Therefore, the oscillations of $\dot p_A(t)$  generates the oscillations of $\dot E_\text{int}(t)$.

Figure~\ref{figapp:zoom} provides a zoom of Fig.~\ref{fig:all} of the main text, containing the plots of the different definitions of entropy production $\sigma^{Es}$, $\sigma^{El}$, $\sigma^{Co}$,  $\sigma^{fp}$, and $\dot I_{A:B}$, in the strong coupling regime.

Finally, Fig.~\ref{fig:RWA} shows the validity of the RWA for the stronger coupling strength chosen for Fig. \ref{fig:all}, namely, $\omega_B/\omega_A = 0.85$, $\Delta/\omega_A = 0.15$, $\omega_A\beta_A = 1.1$, $\omega_A\beta_B = 0.3$, and $g/\omega_A =0.1$. Panel (a) shows the excited population, computed with the Jaynes-Cummings model and with the Rabi model while panel (b) represents the excited population computed with the Jaynes-Cummings model and average of the excited population over the period of fast oscillations of the non-secular terms $2\pi/(\omega_A+\omega_B)$.

\begin{figure}
\begin{center}
    \includegraphics[width=0.45\textwidth]{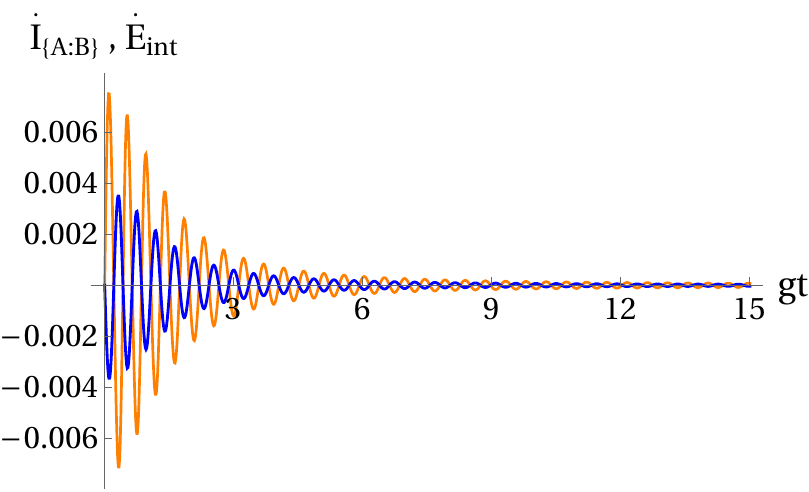}
\end{center}
\caption{Plots of the correlation build-up rate $\dot I_{A:B}$ (in orange) and the time-derivative of interaction energy $\dot E_\text{int}(t) = {\rm Tr}[\dot \rho_{AB}(t) V_{AB}$ (in blue), in weak coupling, with $\omega_B/\omega_A = 0.6$, $\Delta/\omega_A = 0.4$, $\omega_A\beta_A = 1.1$, $\omega_A\beta_B = 0.3$, and $g/\omega_A =0.03 $.}\label{figapp:covsEint}
\end{figure}

\begin{figure}
\begin{center}
    \includegraphics[width=0.45\textwidth]{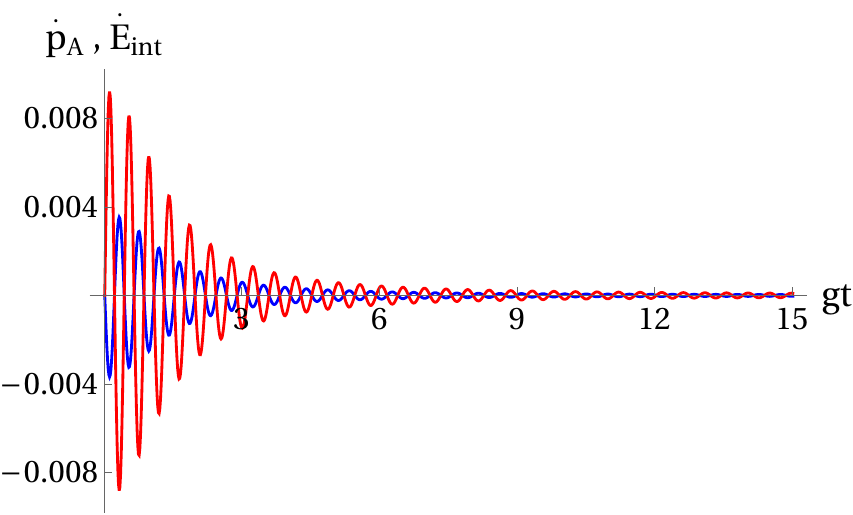}
\end{center}
\caption{Plots of the time-derivative of interaction energy $\dot E_\text{int}(t) = {\rm Tr}[\dot \rho_{AB}(t) V_{AB}$ (in blue) and the time-derivative of the excited populations of $A$ $\dot p_A(t) =\bra 1|\dot\rho_A(t)|1\ket $ (in red), in weak coupling, with $\omega_B/\omega_A = 0.6$, $\Delta/\omega_A = 0.4$, $\omega_A\beta_A = 1.1$, $\omega_A\beta_B = 0.3$, and $g/\omega_A =0.03 $.}\label{figapp:popvsEint}
\end{figure}

\begin{figure}
\begin{center}
    \includegraphics[width=0.45\textwidth]{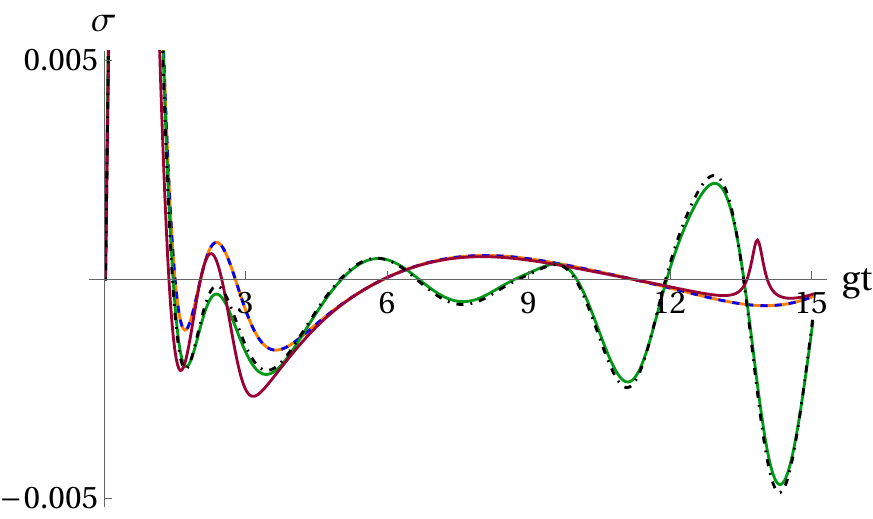}
\end{center}
\caption{Zoom of the plot presented in Fig. \ref{fig:all}, section \ref{sec:EPcomparison}, for the different definitions of entropy production: $\sigma^{Es}$ (orange, thick), $\sigma^{El}$ (green), $\sigma^{Co}$ (purple),  $\sigma^{fp}$ (dashed blue), and $\dot I_{A:B}$ (dot-dashed black), in a strong coupling regime, with $\omega_B/\omega_A = 0.85$, $\Delta/\omega_A = 0.15$, $\omega_A\beta_A = 1.1$, $\omega_A\beta_B = 0.3$, and $g/\omega_A =0.1 $.}\label{figapp:zoom}
\end{figure}

\begin{figure}
\begin{center}
    (a)\includegraphics[width=0.35\textwidth]{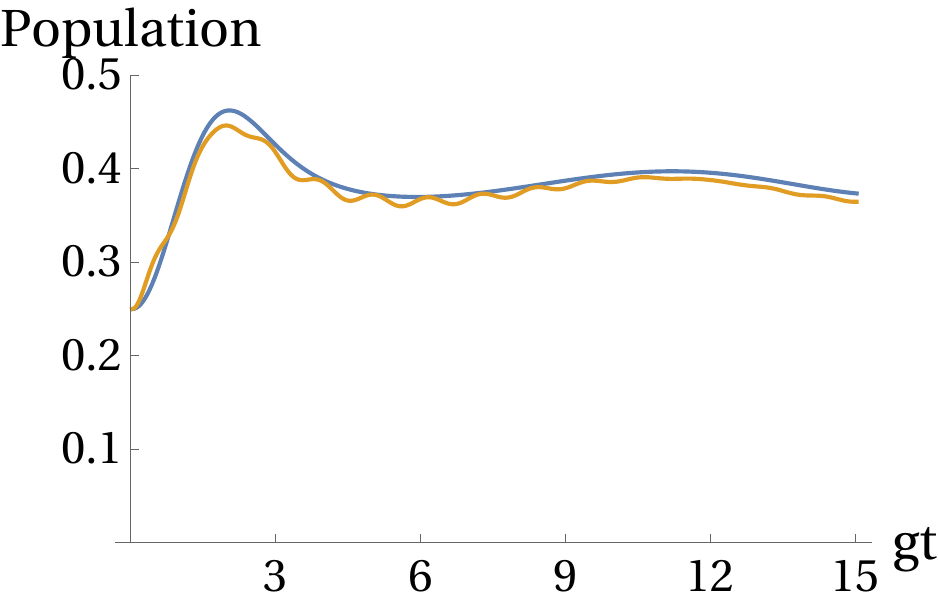}
    (b)\includegraphics[width=0.35\textwidth]{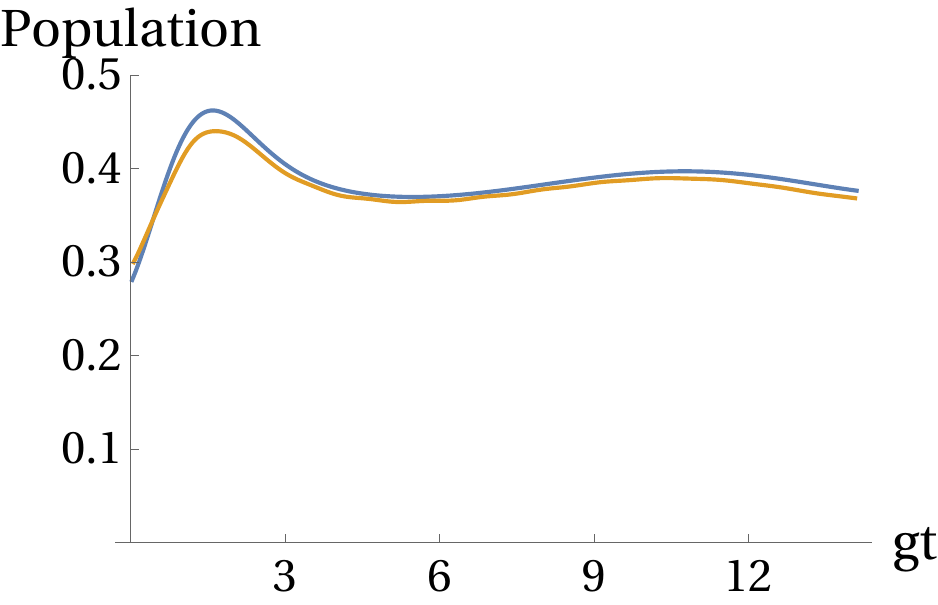}
\end{center}
\caption{(a) Plots of the excited population of system $A$ obtained by the Jaynes-Cummings model (blue line) and by the Rabi model (orange line). (b) Plots of the excited population of $A$ obtained by the Jaynes-Cummings model (blue line) and plot of the excited population given by the Rabi model average over the fast oscillating period $2\pi/(\omega_A+\omega_B)$ (orange line). The parameters used are $\omega_B/\omega_A = 0.85$, $\Delta/\omega_A = 0.15$, $\omega_A\beta_A = 1.1$, $\omega_A\beta_B = 0.3$, and $g/\omega_A =0.1 $ }\label{fig:RWA}
\end{figure}

\section{Comparison between the entropy production, CP-divisibility and the  BLP criterion}\label{app:CPdivBLP}
In this section we present some  results to highligh the relation between  entropy production and CP-divisibility (Fig.~\ref{figapp:CPdiv}) and the BLP criterion (Fig.~\ref{figapp:BLP}). We consider the minimum over all initial states of all entropy production definitions.

Firstly, we recall briefly the meaning of these criteria. Memory effect characterizations consider the divisibility properties of the dynamics. 
 More precisely, the most general dynamics of a quantum system is described by a quantum map~\cite{Jagadish2019}, which are CPTP maps (Completely Positive Trace Preserving maps). A quantum map $\Lambda(t,0)$ describing the time evolution of a quantum system from $0$ to $t$, is said to be CP-divisible (P-divisible) if for all $t' \in[0;t]$, the intertwined map $\Lambda(t,t')$ describing the evolution from $t'$ to $t$ is completely positive (positive). We recall here that the dynamics is CP-divisible at time $t$ if and only if all decay rates in the master equation are positive at this same instant of time $t$~\cite {Rivas2014,Breuer2016, Vega2017,Theret2025}. This criterion is used to obtain the plot in Fig.~\ref{figapp:CPdiv}. 
 
Another well-known memory effect characterization is the BLP criterion~\cite{Breuer2009}, which measures the contractivity of the dynamics, interpreted as a backflow of information. There is no backflow of information (interpreted as the absence of memeory effects) if and only if 
\be\label{eqapp:BLPc}
\frac{d}{dt} ||\Lambda_{t,0} (\rho_1 - \rho_2) ||_1 \leq  0 ~~~~\forall t,
\ee
for any states $\rho_1$ and $\rho_2$, where $||X||_1 := {\rm Tr}[\sqrt{X^\dag X}]$ denotes the trace norm, and $\Lambda_{t,0}$ is the map describing the open dynamics of the system from 0 to $t$. The BLP criterion~\eqref{eqapp:BLPc} means that the map $\Lambda_{t,0}$ is contractive for arbitrary pair of states. For a qubit whose open dynamics is described by the general master equation~\eqref{eq:ME}, it is shown in~\cite{Theret2025} that the BLP criterion \eqref{eqapp:BLPc} is equivalent at any time $t$ to, 
\bea
& & \gamma_1(t) + \gamma_2(t) \geq 0,\nn\\
& & \gamma_3(t) + \frac{\gamma_1(t) + \gamma_2(t)}{2} \geq 0.
\eea
This criterion is used to obtain the plot in Fig. \ref{figapp:BLP}.\\

\begin{figure}
\begin{center}
    \includegraphics[width=0.45\textwidth]{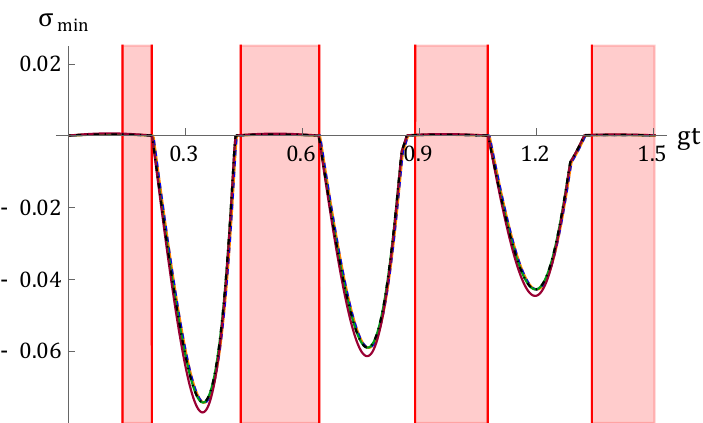}
\end{center}
\caption{Plot of $\sigma^{fp}_\text{min}$ and of the minimum of all entropy production definitions (defined  as in Eq. \eqref{sigmamin}), in the weak coupling regime with $\omega_B/\omega_A = 0.6$, $\Delta/\omega_A = 0.4$,  $\omega_A\beta_B = 0.3$, and $g/\omega_A =0.03 $. The shaded red areas correspond to the intervals of time where the dynamics is CP-divisible.}\label{figapp:CPdiv}
\end{figure}

\begin{figure}
\begin{center}
    \includegraphics[width=0.45\textwidth]{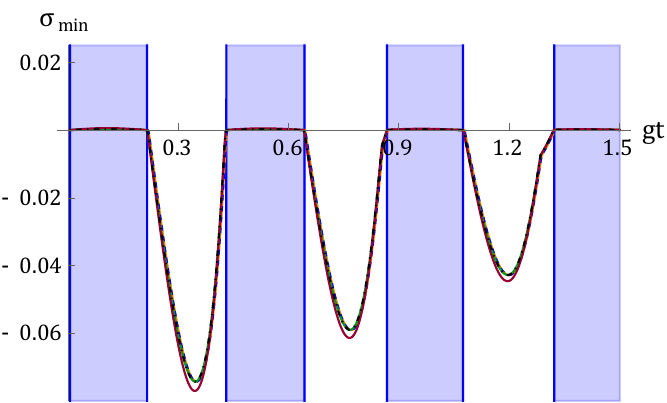}
\end{center}
\caption{Plot of $\sigma^{fp}_\text{min}$ and of the minimum of all entropy production definitions (defined  as in Eq. \eqref{sigmamin}), in the weak coupling regime with $\omega_B/\omega_A = 0.6$, $\Delta/\omega_A = 0.4$,  $\omega_A\beta_B = 0.3$, and $g/\omega_A =0.03 $. The shaded blue areas correspond to the intervals of time where there is no back flow of information according to the BLP criterion (see the text for details).}\label{figapp:BLP}
\end{figure}

\section{Analytical proof of $\sigma^\text{Es} = \sigma^\text{fp}$}\label{app:sigmaid}

From Eq.~\eqref{sigmaEs} and~\eqref{sigmafp}, one can see that we have to compare  $\beta_B^\text{eff}(0) \dot E_B$ and ${\rm Tr}[\dot \rho_A(t)\ln\rho_A^{fp}(t)]$.

\subsection{Explicit expression of $\beta_B^\text{eff}(0)\dot E_B$}
We can re-write $\beta_B^\text{eff}(0)\dot E_B$ as
\bea
\beta_B^\text{eff}(0) \dot E_B &=& \beta_B^\text{eff}(0) {\rm Tr}[\dot \rho_{AB}(t) H_B] \nn\\
&=& -i \beta_B^\text{eff}(0) {\rm Tr}\{[H_{AB}, \rho_{AB}(t) ]H_B\} \nn\\
&=& -i \beta_B^\text{eff}(0) {\rm Tr}\{\rho_{AB}(t) [H_{B}, H_{AB}]\} \nn\\
&=& - i \beta_B^\text{eff}(0) g \omega_B {\rm Tr}[\rho_{AB}(t)(\sigma_-a^\dag - \sigma_+a)]\nn\\
&=& - i \beta_B^\text{eff}(0) g \omega_B \nn\\
&&\times{\rm Tr}[\rho_{AB}(0)U^{\text{Sch}\dag}(t)(\sigma_-a^\dag - \sigma_+a)U^\text{Sch}(t)]\nn\\
&=& - i \beta_B^\text{eff}(0) g \omega_B \nn\\
&&\times{\rm Tr}[\rho_{AB}(0)U^{\dag}(t)(e^{-i t \Delta}\sigma_-a^\dag -e^{i t \Delta} \sigma_+a)U(t)]\nn\\
\eea
with $\Delta := \omega_A - \omega_B$, and $U(t)$ is the join evolution, in the interaction picture, given by~\cite{Smirne2010},
\bea\label{eq:U}
U(t) &=& c(\hat n +1) |1\ket\bra 1| + d(\hat n +1)a |1\ket\bra 0| \nn\\
&& - a^\dag d^\dag(\hat n +1) |0\ket\bra 1| + c^\dag(\hat n) |0\ket\bra 0|,
\eea
with 
\bea\label{eq:c}
c(\hat n) &=& e^{i t\Delta/2} \Big[\cos\left(\sqrt{\Delta^2 + 4g^2\hat n }\frac{t}{2}\right) \nn\\
&&- i \Delta \frac{\sin\left(\sqrt{\Delta^2 + 4g^2\hat n }\frac{t}{2} \right)}{\sqrt{\Delta^2 + 4g^2\hat n } }\Big]
\eea
and 
\be\label{eq:d}
d(\hat n) = - 2i g e^{i t\Delta/2} \frac{\sin\left(\sqrt{\Delta^2 + 4g^2\hat n }\frac{t}{2} \right)}{\sqrt{\Delta^2 + 4g^2\hat n } },
\ee
and $\hat n:= a^\dag a$. Assuming that $\rho_{SB}(0) = \rho_A(0)\otimes\rho_B(0)$, one obtains
\begin{widetext}
\bea
\beta_B^\text{eff}(0) \dot E_B &=& -i g\omega_B\beta_B^\text{eff}(0) \Bigg\{ p_1(0) {\rm Tr}\Big[ \rho_B(0) e^{it\Delta} \Big(c^\dag(\hat n +1)a a^\dag d^\dag(\hat n +1) + a^\dag d^\dag(\hat n+1)a c^\dag(\hat n)\Big) - \text{h.c.}\Big]\nn\\
&&\hspace{2cm} -{\rm Tr}\Big[\rho_B(0)e^{i t \Delta}\Big(a^\dag d^\dag(\hat n +1) a c^\dag(\hat n)\Big)  - \text{h.c.}\Big] \nn\\
&& \hspace{2cm}+ c_{10}(0) {\rm Tr}\Big[ \rho_B(0)  \Big(e^{-it\Delta}c(\hat n )a^\dag c(\hat n +1) +e^{it\Delta} a^\dag d^\dag(\hat n+1)a a^\dag d^\dag(\hat n+1)\Big)\Big]- \text{h.c.}\Bigg\}.
\eea
\end{widetext}
We describe each term in the above expression.
\begin{itemize}
\item Starting with the term $ {\rm Tr}\Big[ \rho_B(0) e^{it\Delta} \Big(c^\dag(\hat n +1)a a^\dag d^\dag(\hat n +1) + a^\dag d^\dag(\hat n+1)a c^\dag(\hat n)\Big) - \text{h.c.}\Big]$:\\
We assume that the initial state of $B$ is diagonal, $\rho_B(0) = \sum_n p_n |n\ket\bra n|$, where $\{|n\ket\}_n$ is the Fock state basis, and $\{p_n\}_n$ is a probability distribution left unspecified for now. We also use the following identity (which can be directly shown from Eq.~\eqref{eq:c} and~\eqref{eq:d}), $c^\dag(\hat n)a^\dag d^\dag(\hat n+1) = b^\dag d^\dag(\hat n+1)c^\dag(\hat n+1)$. Finally, we can show: 
\begin{widetext}
\bea
 &&{\rm Tr}\Big[ \rho_B(0) e^{it\Delta} \Big(c^\dag(\hat n +1)a a^\dag d^\dag(\hat n +1) + a^\dag d^\dag(\hat n+1)a c^\dag(\hat n)\Big) - \text{h.c.}\Big] \nn\\
&&={\rm Tr}\Big[ \rho_B(0) e^{it\Delta} \Big(c^\dag(\hat n +1)a a^\dag d^\dag(\hat n +1) + c^\dag(\hat n)a^\dag d^\dag(\hat n+1)a \Big) - \text{h.c.}\Big]\nn\\
&&={\rm Tr}\Big[ \rho_B(0) e^{it\Delta} \Big(c^\dag(\hat n +1)a a^\dag d^\dag(\hat n +1) + a^\dag d^\dag(\hat n +1) c^\dag(\hat n+1)a \Big) - \text{h.c.}\Big]\nn\\
&&= \sum_{n=0}^\infty p_n e^{it\Delta} \Big[c^*(n+1)(n+1)d^*(n+1) + nd^*(n)c^*(n)\Big] - \text{c.c}\nn\\
&&=2 i \sum_{n=0}^\infty p_n \Big[ (n+1){\rm Im}\Big(e^{it\Delta}c^*(n+1)d^*(n+1)\Big) + n {\rm Im}\Big(e^{it\Delta}c^*(n)d^*(n)\Big)\Big] \nn\\
&&=2 i \sum_{n=1}^\infty n(p_n + p_{n-1})  {\rm Im}\Big(e^{it\Delta}c^*(n)d^*(n)\Big) \nn\\
&&=-2 i \sum_{n=1}^\infty n(p_n + p_{n-1})  {\rm Im}\Big(e^{-it\Delta}c(n)d(n)\Big) \nn\\
&&= -2 i \sum_{n=1}^\infty n(p_n + p_{n-1})    {\rm Im}\Big[-2ig\Big(\cos \Omega_n t/2 - i \frac{ \Delta}{\Omega_n} \sin \Omega_n t/2\Big)\frac{\sin \Omega_n t/2}{\Omega_n}\Big]\nn\\
&&= 4 ig \sum_{n=1}^\infty n(p_n + p_{n-1})  \frac{1}{\Omega_n} \cos \frac{\Omega_n t}{2} \sin \frac{\Omega_n t}{2},
\eea
\end{widetext}
where $c(n)$ and $d(n)$ simply mean that we substitute in Eq.~\eqref{eq:c} and~\eqref{eq:d} the operator $\hat n$ by the integer $n$, and we introduce the frequency $\Omega_n := \sqrt{\Delta^2 + 4g^2n}$.

\item Next, the term $ -{\rm Tr}\Big[\rho_B(0)e^{i t \Delta}\Big(a^\dag d^\dag(\hat n +1) a c^\dag(\hat n)\Big)  - \text{h.c.}\Big]$:\\
Using the same identity as above, we have:
\bea
&-& {\rm Tr}\Big[\rho_B(0)e^{i t \Delta}\Big(a^\dag d^\dag(\hat n +1) a c^\dag(\hat n)\Big)  - \text{h.c.}\Big]\nn\\ &=& -e^{i t \Delta}{\rm Tr}\Big[\rho_B(0) a^\dag d^\dag(\hat n +1) c^\dag(\hat n+1)a \Big] + \text{c.c.}\nn\\
&=& -2i \sum_{n=0}^\infty p_n {\rm Im}\Big(e^{it\Delta} n d^*(n)c^*(n)\Big)\nn\\
&=& 2i \sum_{n=0}^\infty np_n {\rm Im}\Big(e^{-it\Delta}  d(n)c(n)\Big)\nn\\
&=& -4ig \sum_{n=1}^\infty np_n   \frac{1}{\Omega_n} \cos \frac{\Omega_n t}{2} \sin \frac{\Omega_n t}{2}
\eea

\item Finally, the last term ${\rm Tr}\Big[ \rho_B(0)  \Big(e^{-it\Delta}c(\hat n )a^\dag c(\hat n +1) +e^{it\Delta} a^\dag d^\dag(\hat n+1)a a^\dag d^\dag(\hat n+1)\Big)\Big]- \text{h.c.}$ is equal to zero for initial state diagonal state of $B$.
\end{itemize}

Altogether, we find
\begin{widetext}
\bea
\beta_B^\text{eff}(0) \dot E_B &=& 4 g^2\omega_B\beta_B^\text{eff}(0) \sum_{n=1}^\infty \Big[p_1(0) n(p_n+p_{n-1}) - n p_n\Big]   \frac{1}{\Omega_n} \cos \frac{\Omega_n t}{2} \sin \frac{\Omega_n t}{2}\nn\\
&=& \omega_B\beta_B^\text{eff}(0) \sum_{n=1}^\infty \Big[p_1(0) (p_n+p_{n-1}) -  p_n\Big]   \frac{\Omega_n^2 -\Delta^2}{\Omega_n} \cos \frac{\Omega_n t}{2} \sin \frac{\Omega_n t}{2}.
\eea
\end{widetext}
We recall that the above expression is valid for any initial diagonal state $\rho_B(0)$.

\subsection{Explicit expression of ${\rm Tr}[\dot\rho_A(t)\ln \rho_A^{fp}(t)]$}\label{subsec:identity}
We now assume that $B$ is initially in a thermal state $\rho_B(0) = Z^{-1}e^{- H_B/k_BT_B}$. From the structure of the exact master equation~\eqref{eq:ME} followed by $\rho_A(t)$, we can easily deduce that the instantaneous fixed point is
\be
\rho_A^{fp}(t) = p_1^{fp}(t) |1\ket\bra 1| + p_0^{fp}(t)|0\ket\bra 0|,
\ee
 with $p_1^{fp}(t) = \frac{\gamma_1(t)}{\gamma_1(t) + \gamma_2(t)}$ and  $p_0^{fp}(t) = \frac{\gamma_2(t)}{\gamma_1(t) + \gamma_2(t)}$. Then, we have 
\bea
{\rm Tr}[\dot\rho_A(t)\ln \rho_A^{fp}(t)] &=& \dot p_1(t) \ln \frac{\gamma_1(t)}{\gamma_2(t)}.
\eea
From the master equation~\eqref{eq:ME}, one obtains $\dot p_1(t) =  - (\gamma_1(t) + \gamma_2(t))p_1(t) + \gamma_1(t)$. Integrating the differential equation and using Eq.~\eqref{gamma1} and~\eqref{gamma2} of $\gamma_1(t)$ and $\gamma_2(t)$, we obtain $p_1(t) = [\alpha(t) + \beta(t) -1)p_1(0) + 1-\alpha(t)$. One can also show the identity $\gamma_1(t) + \gamma_2(t) = -\frac{\dot\alpha(t) + \dot\beta(t)}{\alpha(t) +\beta(t) -1}$, leading to
\bea
\dot p_1(t) &=& \Big[(\dot \alpha(t) +\dot \beta(t))p_1(0) - \dot\alpha(t)\Big]
\eea
and
\bea
\ln \frac{\gamma_1(t)}{\gamma_2(t)}&=&  \ln \left( \frac{\alpha(t)\dot\beta(t) -\beta(t)\dot\alpha(t) - \dot\beta(t)}{-\alpha(t)\dot\beta(t) +\beta(t)\dot\alpha(t) - \dot\alpha(t)}\right).
\eea
Then, from Eq.~\eqref{eq:alpha} and~\eqref{eq:beta}, we can show:
\bea\label{eq:da}
&&\dot \alpha(t) = \sum_{n=1}^\infty p_n\frac{\Delta^2-\Omega_n^2}{\Omega_n} \cos \frac{\Omega_n t}{2} \sin \frac{\Omega_n t}{2}\\\label{eq:db}
&&\dot \beta(t) = \sum_{n=1}^\infty p_{n-1}\frac{\Delta^2-\Omega_n^2}{\Omega_n} \cos \frac{\Omega_n t}{2} \sin \frac{\Omega_n t}{2},
\eea
from which follows the identities:
\bea
&&\dot \alpha(t) = e^{-\omega_B\beta_B}\dot\beta(t)\nn\\
&&\alpha(t)\dot\beta(t) - \dot\alpha(t)\beta(t)  = Z^{-1} \dot\beta(t)
\eea
with $Z =1/(1-e^{-\omega_B\beta_B})$  and the notation $\beta_B = 1/k_BT_B$. This implies 
\be\label{eq:identities}
\ln \frac{\gamma_1(t)}{\gamma_2(t)} = \ln \left( \frac{\alpha(t)\dot\beta(t) -\beta(t)\dot\alpha(t) - \dot\beta(t)}{-\alpha(t)\dot\beta(t) +\beta(t)\dot\alpha(t) - \dot\alpha(t)}\right) =- \omega_B\beta_B,
\ee
leading to 
\be
{\rm Tr}[\dot\rho_A(t)\ln \rho_A^{fp}(t)] =- \omega_B\beta_B \Big[(\dot \alpha(t) +\dot \beta(t))p_1(0) - \dot\alpha(t)\Big] .
\ee
Using Eq.~\eqref{eq:da} and \eqref{eq:db}, we finally arrive at
\bea
{\rm Tr}[\dot\rho_A(t)\ln \rho_A^{fp}(t)] &=&\omega_B\beta_B\sum_{n=1}^\infty \Big[p_1(0)(p_n+p_{n-1}) -p_n\Big]\nn\\
&&\times \frac{\Omega_n^2-\Delta^2}{\Omega_n}  \cos \frac{\Omega_n t}{2} \sin \frac{\Omega_n t}{2},
\eea
which is precisely the expression found above for $\beta_B\dot E_B$.

In conclusion, for $B$ initially in a thermal state, we indeed have $\sigma^{Es} = \sigma^{fp}$. Note that the relation~\eqref{eq:identities} is only valid for $B$ initially in a thermal state, so the hypothesis of initial thermality of $B$ is necessary to have the equality $\sigma^{Es} = \sigma^{fp}$ (but $A$ can be initially in arbitrary state).


\subsection{Important consequence}\label{app:impcons}
Equation \eqref{eq:identities} reveals that $\gamma_+(t)/\gamma_-(t)$ is indeed constant, which implies that the instantaneous fixed point $\rho_A^{fp}(t)$ is also constant and actually equal to 
\be
\rho_A^{fp} = w\left[H_A,\frac{\omega_B}{\omega_A}\beta_B\right].
\ee
 It is worth mentioning that this result is valid only when $B$ is initially in a thermal state (at inverse temperature $\beta_B$).

\section{P-divisibility implies positivity of $\sigma^{fp}$}\label{app:posep}
The entropy production based on the instantaneous fixed point $ \rho_A^{fp}(t)$ is defined as Eq.~\eqref{sigmafp},
\be
\sigma^{fp} := -{\rm Tr}\{\dot \rho_A(t) [\ln\rho_A(t) - \ln \rho_A^{fp}(t)]\}.
\ee
Note that when $B$ is initially in a thermal state, $ \rho_A^{fp}(t)$ is actually time-independent (see Sec.~\ref{app:impcons}). Introducing $\Lambda(t,0)$, the map describing the evolution of $A$ from 0 to $t$, we can rewrite the above expression as
\bea
\sigma^{fp} &:=&-\lim_{dt \rightarrow 0} \frac{1}{dt}\Big\{ {\rm Tr}\{\Lambda(t+dt,0)\rho_A(0) \Big[\ln\Lambda(t+dt,0)\rho_A(0)\nn\\
&& - \ln \rho_A^{fp}(t)\Big]  - {\rm Tr}\{\rho_A(t) \Big[\ln\rho_A(t) - \ln \rho_A^{fp}(t)\Big]  \Big\},
\eea
 Then, assuming $\Lambda(t,0)$ is P-divisible at time $t+dt$, we have,
\bea
\sigma^{fp} &:=&-\lim_{dt \rightarrow 0} \frac{1}{dt}\Big\{ {\rm Tr}\{\Lambda(t+dt,t)\rho_A(t) \Big[\ln\Lambda(t+dt,t)\rho_A(t)\nn\\
&& - \ln \Lambda(t+dt,t)\rho_A^{fp}(t)\Big]  - {\rm Tr}\{\rho_A(t) \Big[\ln\rho_A(t)\nn\\
&& - \ln \rho_A^{fp}(t)\Big]  \Big\} \geq 0,
\eea
which is always positive by using the contractivity property of the relative entropy under positive maps~\cite{Muller2017}. Importantly, in the above equation, we have used the fact that $\rho_A^{fp}(t)$ is the instantaneous fixed point at time $t$ so that $\Lambda(t+dt,t)\rho_A^{fp}(t) =\rho_A^{fp}(t)$. 

The above proof can be straightforwardly adapted to the fixed point entropy production extended to the Bloch ball, $\tilde\sigma^{fp}[\rho_A]:= -  {\rm Tr} \left\{{\cal L}_t \rho_A\left[ \ln \rho_A - \ln \rho_A^{fp}(t)\right] \right\}$. This implies that the map entropy production $\sigma_\text{map}$ is positive for the instant of times for which the map is P-divisible.

\bibliography{biblio}

\end{document}